\documentclass[12pt]{article}
\usepackage{url}
\usepackage{amsfonts}
\usepackage{amsmath,amssymb,color}
\usepackage{graphicx}
\usepackage{enumitem}
\usepackage{multirow}
\usepackage{bm}
\usepackage{natbib}
\usepackage{verbatim}
\usepackage{float}
\usepackage[toc,page]{appendix}
\usepackage[onehalfspacing]{setspace}
\usepackage{setspace}
\usepackage{adjustbox}
\usepackage{rotating}
\usepackage{subcaption}
\usepackage{amsthm}
\usepackage{booktabs,dcolumn,caption}

\usepackage{array}
\newcolumntype{L}[1]{>{\raggedright\let\newline\\\arraybackslash\hspace{0pt}}p{#1}}
\newcolumntype{C}[1]{>{\centering\let\newline\\\arraybackslash\hspace{0pt}}p{#1}}
\newcolumntype{R}[1]{>{\raggedleft\let\newline\\\arraybackslash\hspace{0pt}}p{#1}}

\newtheorem{theorem}{Theorem}

\newtheorem{corollary}{Corollary}

\newtheorem{example}{Example}

\newtheorem{lemma}{Lemma}

\newtheorem{proposition}{Proposition}
\newtheorem{remark}{Remark}

\newtheorem{assumption}{Assumption}

\newcommand{\mM}{{\mathcal M}}
\newcommand{\MM}{{M}}
\newcommand{\YY}{Y}
\newcommand{\ZZ}{Z}

\title{Determining the Number of Factors in High-dimensional
Generalised Latent Factor Models}

\author{Yunxiao Chen,\\
London School of Economics and Political Science\\
  Xiaoou Li,\\
  University of Minnesota}
\date{}

\begin{document}
\maketitle

\doublespacing

\begin{abstract}
As a generalization of the classical linear factor model, generalized latent factor models are useful for analyzing multivariate data of different types, including binary choices and counts. This paper proposes an information criterion to determine the number of factors in generalized latent factor models.   The consistency of the proposed information criterion is established under a high-dimensional setting where both the sample size and the number of manifest  variables grow to infinity, and data may have many missing values. An error bound is established for the parameter estimates, which plays an important role in establishing the consistency of the proposed information criterion. This error bound improves several existing results and may be of independent theoretical interest. We evaluate the proposed method by a simulation study and an application to Eysenck's personality questionnaire.
\end{abstract}

\noindent
Keywords: 
Generalized latent factor model;  Joint maximum likelihood estimator; High-dimensional data; Information criteria; Selection consistency

\section{Introduction}\label{sec:intro}
Factor analysis is a popular method in social and behavioral sciences, including psychology, economics, and marketing \citep{bartholomew2011latent}.  It uses a relatively small number of factors to model the variation in a large number of observable variables,  often known as manifest variables. For example, in psychological science, manifest variables may correspond to personality questionnaire items for which factors are often interpreted as personality traits.
Multivariate data in social and behavioral sciences often involve categorical or count variables, for which the classical linear factor model may not be suitable. Generalized latent factor models \citep{skrondal2004generalized,chen2019structured} provide a flexible framework for more types of data by combining generalized linear models and factor analysis. Specifically, item response theory models \citep{embretson2000item,reckase2009multidimensional}, which are widely used in psychological measurement and educational testing, can be viewed as   special cases of generalized latent factor models. The generalized latent factor models are also closely related to several low-rank models for count data \citep{liu2018pca,robin2019low,mcrae2019low} and mixed data \citep{collins2002generalization,robin2020main} that make
similar probabilistic assumptions, though these works do not pursue interpretations from the factor analysis perspective.

Factor analysis is often used in an exploratory manner for generating scientific hypotheses. In this case, which is known as exploratory factor analysis, the number of factors and the corresponding loading structure are unknown and need to be learned from data. Quite a few methods have been proposed for determining the number of factors in linear factor models, including eigenvalue-based criteria \citep{kaiser1960application,cattell1966scree, onatski2010determining,ahn2013eigenvalue}, information criteria \citep{bai2002determining,bai2018consistency,choi2019model}, cross-validation \citep{owen2016bi}, and parallel analysis \citep{horn1965rationale,buja1992remarks,dobriban2019deterministic}. However, fewer methods are available for determining the number of factors in generalized latent factor models, and statistical theory remains to be developed, especially under a high-dimensional setting when the sample size and the number of manifest variables are large.

%When statistical inference is carried out by marginal maximum likelihood estimation under a low-dimensional setting,
%the classical information criteria, such as the Bayesian information criterion, can still be used. Under a high-dimensional setting, \cite{chen2019joint} adopted a cross-validation approach and \cite{zhang2020note} proposed
%a scree-plot-based approach.

%for item response data that are usually categorical.

Traditionally, statistical inference of generalized latent factor models is typically carried out based on a marginal likelihood function \citep{bock1981marginal,skrondal2004generalized}, in which latent factors are treated as random variables and are integrated out from the likelihood function. However, for high-dimensional data involving large numbers of observations, manifest variables and factors, marginal-likelihood-based inference tends to suffer from a high computational burden and thus may not always be feasible. In that case,  a joint likelihood function  that treats factors as fixed model parameters  may be a good alternative\citep{chen2019joint,chen2019structured,zhu2016personalized}. Specifically, a joint maximum likelihood estimator is proposed in \cite{chen2019joint,chen2019structured} that is easy to compute and   also statistically optimal in the minimax sense when both the sample size and the number of manifest variables grow to infinity. With a diverging number of parameters in the joint likelihood function, the classical information criteria, such as the Akaike information criterion \citep[AIC;][]{akaike1974new} and the Bayesian information criterion \citep[BIC;][]{schwarz1978estimating}, may no longer be suitable.

%The statistical inference of generalized latent factor models is usually carried out under an empirical Bayes framework.
%Under this framework, parameter estimation and model selection are based on
%a marginal likelihood function \citep{bock1981marginal,skrondal2004generalized}, in which latent factors are treated as random variables and integrated out. Under a low-dimensional setting where the numbers of manifest variables and factors are small, maximising the marginal likelihood function is computationally easy by using an expectation-maximization (EM) algorithm \citep{dempster1977maximum,bock1981marginal}, and statistical inference (for example,  testing hypothesis and model selection) can be carried out using the classical theory for parametric maximum likelihood estimation. However, this empirical Bayes framework may no longer be suitable under a high-dimensional setting. Specifically, when there are many factors, maximising the marginal likelihood function by the EM algorithm is computationally intensive, due to the need of calculating high-dimensional numerical integrals. Although the computation can be speeded up by some stochastic version of the EM algorithm that replaces the numerical integrals by Markov chain Monte Carlo sampling, the implementation of these stochastic algorithms is usually not easy. Moreover,

This paper proposes a joint-likelihood-based information criterion (JIC) for determining the number of factors in generalized latent factor models. The proposed criterion is suitable for high-dimensional data with large numbers of observations and manifest variables and can be used even when data contain many missing values. Under a very general setting, we prove the consistency of the proposed JIC when both the numbers of samples and manifest variables grow to infinity. Specifically, the missing entries are allowed to be non-uniformly distributed in the data matrix, and their proportion is allowed to grow to one, i.e., the proportion of observable entries is allowed to decay to zero. An error bound for the joint maximum likelihood estimator is established under a general setting where the data entries can be non-uniformly missing, and the number of factors can grow to infinity. This error bound substantially extends the existing results on the estimation of generalized latent factor models and related models, including \cite{cai2013max}, \cite{davenport20141},  \cite{bhaskar20151}, \cite{ni2016optimal}, and \cite{chen2019structured}. Simulation shows that the proposed JIC has good finite sample performance under different settings, and an application to the revised Eysenck's personality questionnaire \citep{eysenck1985revised} finds three factors, which confirms the design of this personality survey.

% and the number of factors are all large.

\section{Joint-likelihood-based Information Criterion}

\subsection{Generalized Latent Factor Models}

We consider multivariate data involving $N$ individuals and $J$ manifest variables. Let $y_{ij}$ be a random variable that denotes the $i$th individual's value on the $j$th manifest variable. Factor models assume that each individual is associated with $K$ latent factors, denoted by a vector $F_i = (f_{i1}, ..., f_{iK})^T$. We assume that the distribution of $y_{ij}$  given $F_i$ follows an exponential family distribution with natural parameter
$d_j + A_j^T F_i,$
and possibly a scale parameter $\phi$ that is also known as a dispersion parameter, where $d_j$ and $A_j = (a_{j1}, ..., a_{jK})^T$ are manifest-variable-specific parameters. Specifically, $d_j$ can be viewed as an intercept parameter, and $a_{jk}$ is known as a loading parameter.
More precisely, the probability density/mass function for $y_{ij}$ takes the form
\begin{equation}\label{eq:models}
g(y|A_j, d_j, F_i, \phi) = \exp\left\{\frac{y(d_j + A_j^T F_i) - b(d_j + A_j^T F_i)}{\phi} + c(y, \phi)\right\},
\end{equation}
where $b$ and $c$ are pre-specified functions that depend on  the exponential
family distribution. Given all the person- and manifest-variable-specific parameters, data $y_{ij}$, $i = 1, ..., N$, $j = 1, ..., J$, are assumed to be independent.
In particular, linear factor models for continuous data,
logistic factor model for binary data, and Poisson factor model for counts, are special cases of model \eqref{eq:models}.
We present the logistic and Poisson models as two examples, while pointing out that \eqref{eq:models} also includes linear factor models as a special case when the exponential family distribution is chosen to be a Gaussian distribution.

\begin{example}\label{emp:binary}
When data are binary, \eqref{eq:models} leads to a logistic model. That is, by letting
$b(d_j + A_j^T F_i) = \log\{1+\exp(d_j + A_j^T F_i)\}$, $\phi = 1$, and $c(y, \phi) = 0$, \eqref{eq:models} implies that
$y_{ij}$ follows a Bernoulli distribution with success probability $\exp(d_j + A_j^T F_i)/\{1+\exp(d_j + A_j^T F_i)\}$. This model is known as the multi-dimensional two-parameter logistic model \citep{reckase2009multidimensional} that is widely used in educational testing and psychological measurement.
\end{example}
\begin{example}\label{emp:poisson}
For count data,  \eqref{eq:models} leads to a  Poisson model by letting $b(d_j + A_j^T F_i) = \exp(d_j + A_j^T F_i)$, $\phi  = 1$, and $c(y, \phi ) = -\log(y!)$. Then $y_{ij}$ follows a Poisson distribution with intensity $\exp(d_j + A_j^T F_i)$.  This model is known as the Poisson factor model for count data \citep{wedel2003factor}.
\end{example}

We further take missing data into account under an ignorable missingness assumption. Let $\omega_{ij}$ be a binary random variable, indicating the missingness of $y_{ij}$. Specifically, $\omega_{ij} = 1$ means that $y_{ij}$ is observed, and $\omega_{ij} = 0$ if $y_{ij}$ is missing.
It is assumed that, given all the person- and manifest-variable-specific parameters,  the missing indicators $\omega_{ij}$, $i = 1, ..., N, j = 1, ..., J$, are independent of each other, and are also independent of data $y_{ij}$.
The same missing data setting is adopted in \cite{cai2013max} for a 1-bit matrix completion problem and \cite{zhu2016personalized}
for collaborative filtering. For nonignorable missing data, one may need to model the distribution of $\omega_{ij}$
given $y_{ij}$, $F_i$, $A_j$, and $d_j$. See \cite{little2019statistical} for more discussions on nonignorable missingness. For the ease of explanation, in what follows, we assume the dispersion parameter $\phi>0$ is known and does not change with $N$ and $J$. Our theoretical development below can be extended to the case when $\phi$ is unknown; see Remark~\ref{rmk:dispersion} below for a discussion.

\subsection{Proposed Information Criterion}\label{section:method}

Under the above setting for generalized latent factor models, the log-likelihood function for observed data takes the form
\begin{equation}\label{eq:jointlik}
l_K(F_1,..., F_N, A_1, d_1, ..., A_J, d_J) = \sum_{\omega_{ij}=1} \log  g(y_{ij}|A_j, d_j, F_i, \phi).
\end{equation}
%where $\Omega = \{(i,j): \omega_{ij} = 1, i = 1, ..., N, j = 1, ..., J\}$ denotes the index set for the observed data entries.
%\xc{removed definition of $\Omega$ as it is not used here.}
Note that a subscript $K$ is added to the likelihood function to emphasize the  number of factors in the current model.

For exploratory factor analysis, we consider the following constrained joint maximum likelihood estimator as proposed in \cite{chen2019joint,chen2019structured}
\begin{equation}\label{eq:jml}
\begin{aligned}
(\hat F_1,..., \hat F_N, \hat A_1, \hat d_1, ..., \hat A_J, \hat d_J) \in ~~ \arg\max &~~ l_K(F_1,..., F_N, A_1, d_1, ..., A_J, d_J),\\
s.t. &~~ (\Vert F_i \Vert^2 + 1)^{\frac{1}{2}} \leq C, i = 1, ..., N, \\
     &~~ (d_j^2 + \Vert A_j \Vert^2)^{\frac{1}{2}}  \leq C, j = 1, ..., J,
\end{aligned}
\end{equation}
where $\Vert \cdot\Vert$ denotes the standard Euclidian norm. Here $C$ is a reasonably large constant to ensure that a finite solution to \eqref{eq:jml} exists and satisfies certain regularity conditions.
% existence of a finite solution to \eqref{eq:jml}.
%{\color{red} (I replaced the constraint for $F_i$ to $(\|F_i\|^2+1)^{1/2}\leq C$)}

As there is no further constraint imposed under the exploratory factor analysis setting, the solution to \eqref{eq:jml} is not unique.  This indeterminacy of the solution will not be an issue when determining the number of factors, since the proposed JIC only depends on the log-likelihood function value rather than the value of the specific parameters.
The computation of \eqref{eq:jml} can be done by an alternating maximization algorithm
which has good convergence properties according to numerical experiments \citep{chen2019joint,chen2019structured}, even though \eqref{eq:jml} is a non-convex optimization problem. See Appendix~\ref{app:optimization} of the online supplement for further discussions on the computation of \eqref{eq:jml} and the choice of the constraint constant $C$.

Let $n$ be the number of observed data entries,
i.e., $$n = \sum_{i=1}^N \sum_{j=1}^J \omega_{ij}.$$
The proposed JIC takes the form
$$\text{JIC}(K) = -2 \hat l_K + v(n,N,J,K),$$
where $\hat l_K = l_K(\hat F_1,..., \hat F_N, \hat A_1, \hat d_1, ..., \hat A_J, \hat d_J) $
with $\hat F_i$, $\hat A_j$, and $\hat d_j$ given by \eqref{eq:jml}, and
$v(n,N,J,K)$  is a penalty term depending on $n$, $N$, $J$, and $K$. We choose $\hat K$ that minimizes $\text{JIC}(K)$.

As
will be shown in Section~\ref{section:theory}, the consistency of $\hat K$ can be guaranteed under a wide range of choices of $v(n,N,J,K)$.
In practice, we suggest to use
\begin{equation}\label{eq:suggested}
 v(n,N,J,K) = {K}({N \vee J})\log\{n/(N\vee J)\},
\end{equation}
where $N \vee J$  denotes the maximum  of $N$ and $J$. When there is no missing data, i.e., $n = NJ$, then \eqref{eq:suggested} becomes  $v(n,N,J,K) = {K}({N \vee J})\log(N\wedge J)$, where $N \wedge J$  denotes the minimum  of $N$ and $J$.
%This is a conservative choice in the sense that it is almost the smallest choice one can take to ensure the
The advantage of this choice will be explained in Section~\ref{section:theory}.

%Specifically, $v(n,N,J,K)$
\section{Theoretical Results}\label{section:theory}
We start with the definition of several useful quantities. Let $p_{ij}=\Pr(\omega_{ij}=1)$ be the sampling weight for $y_{ij}$ and $p_{\min}=\min_{1\leq i\leq N,1\leq j\leq J}p_{ij}$ be their minimum. Also let $n^*=\sum_{i=1}^N\sum_{j=1}^J p_{ij}$, $n^*_{i\cdot}=\sum_{j=1}^J p_{ij}$, and $n^*_{\cdot j}=\sum_{i=1}^N p_{ij}$ be the expected number of observations in the entire data matrix, each row and each column, respectively. Let $p_{\max}= (J^{-1}\max_{1\leq i\leq N}n^*_{i\cdot})\vee (N^{-1}\max_{1\leq j\leq J}n^*_{\cdot j})$ be the maximum average sampling weights for different columns and rows.
Let $m^*_{ij}=d^*_j + (A^*_j)^T F^*_i$ be the true natural parameter for $y_{ij}$, and let $M^*=(m^*_{ij})_{1\leq i\leq N,1\leq j\leq J}$. We also denote $\hat{M}=(\hat{d}_j + \hat{A}_j^T \hat{F}_i)_{N\times J}$ to be the corresponding estimator of $M$ obtained from \eqref{eq:jml}. To emphasize the dependence on the number of factors, we use $\hat{M}^{(K)}$ to denote the estimator when assuming $K$ factors in the model. Let $K_{\max}$ denote the maximum number of factors considered in the model selection process and let $K^*$ be the true number of factors.

The following two assumptions are made throughout the paper.
\begin{assumption}\label{assump:support-b}
  For all $x\in [-2C^2,2C^2]$, $b(x)<\infty$.
\end{assumption}
\begin{assumption}\label{assump:true}
  The true model parameters $F_i^*$, $A_j^*$, and $d_j^*$ satisfy the constraint in \eqref{eq:jml}. That is, $(\Vert F^*_i \Vert^2 + 1)^{\frac{1}{2}} \leq C$ and $\{(d^*_j)^2 + \Vert A^*_j \Vert^2\}^{\frac{1}{2}}  \leq C$, for all $i$ and $j$.
\end{assumption}
In the rest of the section, we will first present error bounds for the joint maximum likelihood estimator, and then present conditions on $v(n,N,J,K)$ that guarantee consistent model selection.
\begin{theorem}\label{thm:error-bound-simple}
Assume
%Assumptions~\ref{assump:support-b} and \ref{assump:true} hold and
$n^*/(\log n^*)^2\geq (N\wedge J)\log(N+J)$ and the true number of factors satisfies $1\leq K^*\leq K_{\max}$.
 % {\color{red} looks weird. Need check. If true, it improve the requirement on $n^*$ when $N\gg J$.}
 Then, there is a finite constant $\kappa$ depending on $p_{\max}/p_{\min}$, $C$, $\phi$, the function $b$ and independent of $K_{\max}$, $N$, $J$ and $n^*$, such that with probability at least $1-(n^*)^{-1}-2(N+J)^{-1}$,
  \begin{equation}\label{eq:error-bound-simple}
   \max_{ K^*\leq K\leq K_{\max}}\Big\{(NJ)^{-1/2}\big\|\hat{M}^{(K)}-M^*\big\|_F\Big\}\leq \kappa \Big\{\frac{K_{\max} (N\vee J)}{n^*}\Big\}^{1/2}.
  \end{equation}
  In particular, if $K^*$ is known, then we have $(NJ)^{-1/2}\|\hat{M}^{(K^*)}-M^*\|_F\leq \kappa \big\{{K^* (N\vee J)}/{n^*}\big\}^{1/2}$.
\end{theorem}
The upper bound established in Theorem 1 is sharp, in the sense that the following lower bound holds under mild conditions.
% {\color{red} From this theorem, I would think that $n* >> N\vee J$ is needed. However, from theorem 3 below, it roughly requires that $n* >> N\wedge J$? Don't we need the scaled Frobenius norm to go to zero? If we do need $n* >> N\vee J$, then I would like to change the suggested penalty to
% $v(n,N,J,K) = {K}({N \vee J})\log(n/(N\vee J)$. Although it only changes the log term, it should matter in the simulation.
%  }
%{\color{red} (what format is allowed?)}

\begin{proposition}[Lower bound]\label{prop:lower-bound}
Assume $(K^*)^2(J+N)\leq n^*$. Then, there are  constants $\kappa, N_0,J_0>0$, such that for any $N\geq N_0$, $J\geq J_0$, and any estimator $\bar{M}$,
  \begin{equation}
    \sup_{M^*\in\mathcal{G}}\Pr\Big[(NJ)^{-1/2}\big\|\bar{M}-M^*\big\|_F\geq \kappa^{-1} \{{K^* (N\vee J)}/{n^*}\}^{1/2}\Big]\geq \frac{1}{2},
  \end{equation}
  where $\mathcal{G}=\{M^*=(m^*_{ij}):   F^*_i,A^*_j\in R^{K^*}, (\Vert F^*_i \Vert^2 + 1)^{\frac{1}{2}} \leq C, \{(d^*_j)^2 + \Vert A^*_j \Vert^2\}^{\frac{1}{2}}\leq C \text{ for all }i, j\}$ denotes the parameter space. {Here, $\kappa$ is a constant that depends on $p_{\max}/p_{\min}$, $C$, $\phi$, the function $b$, and is independent of $K^*$.
  It is possibly different from the $\kappa$  in Theorem~\ref{thm:error-bound-simple}. }
   %{\color{red} Should we replace $\kappa$ by a different constant? }
\end{proposition}
%m^*_{ij}=d^*_j+(A^*_j)^TF^*_i ,

We make a few remarks on Theorem~\ref{thm:error-bound-simple}.
\begin{remark}
  It is well-known that in exploratory factor analysis, the factors $F_1,\cdots, F_N$ are not identifiable due to rotational indeterminacy, while $m_{ij}$s are identifiable. Thus, we establish error bounds for estimating the matrix $M$ rather than those of $F_i$s and $A_j$s. If additional design information is available and a confirmatory generalized latent factor model is used, then the methods described in Section~\ref{section:method} and theoretical results  in Theorem~\ref{thm:error-bound-simple} can be extended to establish error bounds for $F_i$s following a similar strategy as in \cite{chen2019structured}.

  The key assumption for Theorem~\ref{thm:error-bound-simple} to hold is that both $M^*$ and $\hat M$ are low-rank matrices. It can be easily generalized to other low rank models beyond the current generalized latent factor model, including the low-rank interaction model proposed in \cite{robin2019low}. For example, one may parameterize $m_{ij} = {d}_j + {A}_j^T {F}_i + d_i^\dagger$, where $d_i^\dagger$ is a person-specific intercept term.

  % still holds when certain constraints are imposed on both $M^*$ and its estimate. For example,

\end{remark}

\begin{remark}\label{rmk:2}
  The error bound \eqref{eq:error-bound-simple}
  %$(NJ)^{-1/2}\|\hat{M}^{(K^*)}-M\|_F=O_p\big(\{\frac{K^* (N\vee J)}{n^*}\}^{1/2}\big)$
  %established in Theorem~\ref{thm:error-bound-simple}
  improves several recent results on low-rank matrix estimation and completion. For example, when $n^*=o\big\{(N\wedge J)^2\big\}$, it improves the error rate $O_p\big[\{{(N\vee J)}{(n^*)^{-1}} + {NJ}{(n^*)^{-3/2} }\}^{1/2}\big]$ in \cite{chen2019structured}, where a fixed $K^*$ and uniform sampling, i.e., $p_{\max}=p_{\min}$, are assumed.
  Other examples include \cite{ni2016optimal}  and \cite{bhaskar20151}, where the error rates are shown to be $O_p[\{{K^*(N\vee J)\log(N+J)}{(n^*)^{-1}}\}^{1/2}]$ and $O_p\{ K^*(N\vee J)^{1/2}(n^*)^{-1/2} + (N\vee J)^{3}(N\wedge J)^{1/2} (K^*)^{3/2}(n^*)^{-2}\}$, respectively, assuming binary data. The error estimate \eqref{eq:error-bound-simple} is also smaller than the optimal rate $\{ {K^*(N\vee J)}{(n^*)^{-1}}\}^{1/4}$ for approximate low rank matrix completion \citep{cai2013max,davenport20141}, which is expected as the parameter space in these works, which consists of
  nuclear-norm constrained matrices, is larger than that of our setting.
  Several technical tools are used to obtain the improved error bound including a sharp bound on the spectral norm of random matrices that extends a recent result in \cite{bandeira2016sharp} and an upper bound of singular values of  Hadamard products of low rank matrices based on a result established in \cite{horn1995norm}.
  %and \cite{li2020network,sussman2013consistent} for network data analysis where the error rates of their estimators are proved to be $O()$
\end{remark}
% \begin{remark}
% We assume the support of $b(\cdot)$ covers $[-2C^2,2C^2]$ so that the optimization \eqref{eq:jml} is well-defined. This assumption can be relaxed by replacing $[-2C^2,2C^2]$ with $[-(1+\epsilon)C^2,(1+\epsilon)C^2]$ for any positive number $\epsilon$ and requiring the constant $\kappa$ to be $\epsilon$-dependent.
% \end{remark}
Note that the constant $\kappa$ in Theorem 1 depends on $p_{\max}/p_{\min}$. Thus, it is most useful when $p_{\max}/p_{\min}$ is bounded by a finite constant that is independent of $N$ and $J$. In this case, the asymptotic error rate is similar between a uniform sampling and a weighted sampling. In the case where the sampling scheme is far from a uniform sampling, the next theorem provides a finite sample error bound.
\begin{theorem}\label{thm:finite-bound}
Let
$  \kappa_{2C^2}=\sup_{|x|\leq 2C^2}b''(x)$, $\delta_{C^2}=\frac{1}{2}\inf_{|x|\leq C^2}b''(x)$,
%\xc{I overlooked the $1/2$ in delta earlier. Removing it will require modifying a lot of places in the proof. That's why $\frac{1}{2}$ appear in the definition....}
 $\kappa_{1,b,C,\phi}=8\delta_{C^2}^{-1}(\phi\kappa_{2C^2})^{1/2}+16C^2$ and $\kappa_{2,b,C,\phi}=(\phi/C^2)\vee (\phi\kappa_{2C^2})^{1/2}$. Then, there exists a universal constant $c$ such that with probability at least $1-2(N+J)^{-1}-(n^*)^{-1}$,
%{\color{red} absorb constant for $K$ in $\kappa$}
\begin{equation}\label{eq:finite-bound-thm}
  \begin{split}
        &\max_{K^*\leq K\leq K_{\max}}\big\|\hat{M}^{(K)}-M^*\big\|_F\\
        \leq & p_{\min}^{-1}K_{\max}^{1/2}\{ \kappa_{1,b,C,\phi}(\max_{i} n_{i\cdot}^*)^{1/2}\vee (\max_{j} n_{\cdot j}^*)^{1/2} + c(\kappa_{2,b,C,\phi}\log n^*+2C^2)\log^{1/2}(N+J) \}
  \end{split}
  \end{equation}
  for all $N\geq 1, J\geq 1$, $n^*\geq 6$ and $ K_{\max}\geq K^*\geq 1$.
\end{theorem}
% {\color{red}$K_{\max}+1$ is replaced by $2K_{\max}$ and is absorbed in $\kappa_{1,b,C}$. Need to change the proof correspondingly.}
% \begin{remark}
% We elaborate on the terms in the the upper bound \eqref{eq:finite-bound-thm}. First, the term $K_{\max}+1$ matches the largest rank of $M$ considered in the models. If the intercepts $d_1,\cdots,d_J$ are known to be zero, then $K_{\max}+1$ can be replaced by $K_{\max}$ and \eqref{eq:finite-bound-thm} still holds. Second, when
% %We note that $K_{\max}+1$ instead of $K_{\max}$ appear in the upper bound because of the additional rank brought by the intercepts $d_1,\cdots,d_J$.
% \end{remark}
  \begin{remark}
  Theorem~\ref{thm:finite-bound} provides a finite sample error bound for the joint maximum likelihood estimator when the number of factors is known to be no greater than $K_{\max}$. It extends Theorem~\ref{thm:error-bound-simple} in several aspects. First, the constants $  \kappa_{2C^2}$, $\delta_{C^2}$, $\kappa_{1,b,C,\phi}$ and $\kappa_{2,bC,\phi}$ are made explicit in Theorem~\ref{thm:finite-bound}. In addition, it allows the missingness pattern to be far from uniform sampling. To see this, consider the case where $J=N^{\alpha}$, $p_{\min}=N^{-\beta}$, and $p_{\max}/p_{\min}\leq N^{\gamma}$ with $\alpha\in (0,1]$, $\beta\in [0,\alpha)$,  $\gamma\in [0,\beta]$ and $C$ is fixed. Roughly, a larger $\gamma$ suggests a more imbalanced sampling scheme.
  Then, Theorem~\ref{thm:finite-bound} implies $(NJ)^{-1/2}\|\hat{M}^{(K^*)}-M^*\|_F= O_p\{N^{(\beta+\gamma-\alpha)/2}(K^*)^{1/2}\}$. Thus, if $\gamma<\alpha-\beta$ and $K^*=o(N^{\alpha-\beta-\gamma})$, the estimator $\hat{M}^{(K^*)}$ is consistent in the sense that the scaled Frobenius norm $(NJ)^{-1/2}\|\hat{M}^{(K^*)}-M^*\|_F$ decays to zero.
  \end{remark}
Let $u(n,N, J,K)=v(n,N,J,K)-v(n,N,J,K-1)$, and let $\sigma_{1}(M^*)\geq \sigma_{2}(M^*)\geq \cdots\geq \sigma_{K^*+1}(M^*)$ be the non-zero singular values of $M^*$. Note that due to the inclusion of the intercept term $d_j$, a non-degenerate $M^*$ is of rank $K^*+1$.
The next theorem provides sufficient conditions on $u(n,N,J,K)$ for consistent model selection.
\begin{theorem}\label{thm:IC-simplified}
Consider the following asymptotic regime as $N,J\to\infty$,
\begin{equation}\label{eq:IC-simple-asym-reg}
   C=O(1), p_{\min}^{-1}=O(1), \text{ and } K^*=O(1).
 \end{equation}
If the function $u$ satisfies
\begin{equation}\label{eq:simplified-ic-condition-u}
  u(n,N,J,K)=o\big\{\sigma_{K^*+1}^2(M^*)\big\} \text{ and } N\vee J = o\big\{u(n,N,J,K)\big\} \text{ uniformly in } K
 \text{ as } N,J\to\infty,
\end{equation}
then, $\lim_{N,J\to\infty}\Pr(\hat{K}=K^*)=1.$
\end{theorem}
\begin{remark}
We elaborate on the asymptotic regime \eqref{eq:IC-simple-asym-reg} and the conditions on $u(n,N,J,K)$ in \eqref{eq:simplified-ic-condition-u}. First, $C=O(1)$ and $K^*=O(1)$ require that $C$ and the number of factors are bounded as $N$ and $J$ grow. Second, $p_{\min}^{-1}=O(1)$ suggests that the missingness pattern is similar to uniform sampling with $n^*$ growing at the order of $NJ$. Third, $u(n,N,J,K)=o\{\sigma_{K^*+1}^2(M^*)\}$ requires that  $u(n,N,J,K)$ is smaller than the gap between non-zero singular values and zero-singular values of $M^*$. Under this requirement, the probability of underselcting the number of factors is small.
%the JIC will be smaller when we under-select the number of factors, compared with that of the correctly selected model.
Fourth, $N\vee J = o\{u(n,N, J,K)\}$ requires that $u(n,N,J,K)$ grows in a faster speed than $N\vee J$. This requirement guarantees that with high probability, we do not over-select the number of factors. Fifth, $n$ is random when there are missing data, and thus $u(n,N,J,K)$ may also be random.
In this theorem, we do not allow  $u(n,N,J,K)$ to be random as implicitly required by  condition \eqref{eq:simplified-ic-condition-u}.
 %We rule out this case of random $u(n,N,J,K)$ implicitly by \eqref{eq:simplified-ic-condition-u} in Theorem~\ref{thm:IC-simplified}.
 A general result allowing a random $u(n,N,J,K)$ is given in
Theorem~\ref{thm:IC-general} below.
\end{remark}

\begin{remark}\label{rmk:5}
We provide further explanations on the requirements of $u(n,N,J,K)=o\{\sigma_{K^*+1}^2(M^*)\}$ and $N\vee J = o\big\{u(n,N,J,K)\big\}$. First, $\sigma_{K^*+1}^2(M^*)$ is the smallest non-zero singular value of $M^*$ that measures the strength of the factors. Under the conditions of Theorem~\ref{thm:IC-simplified}, $\sigma_{K^*+1}^2(M^*)/2(\hat l_{K} - \hat l_{K-1}) = O_p(1)$, when $K \leq K^*$. By letting
$u(n,N,J,K)= o\{\sigma_{K^*+1}^2(M^*)\}$, it is guaranteed that, when $K \leq K^*$, $\text{JIC}(K) - \text{JIC}(K-1)= -2(\hat l_{K} - \hat l_{K-1}) + u(n,N,J,K) < 0$  with probability tending to 1. It thus avoids under-selection. Second, under the conditions of Theorem~\ref{thm:IC-simplified}, $2(\hat l_{K} - \hat l_{K-1})  = O_p (N \vee J)$ {for each fixed} $K \geq K^* +1$, i.e., when both models are correctly specified. When $N\vee J = o\big\{u(n,N,J,K)\big\}$ and  $K \geq K^* +1$,  $\text{JIC}(K) - \text{JIC}(K-1)= -2(\hat l_{K} - \hat l_{K-1}) + u(n,N,J,K) >  0$  with probability tending to 1. This avoids over-selection. Finally, the two requirements also imply that selection consistency can only be guaranteed when $N\vee J = o\{\sigma_{K^*+1}^2(M^*)\}$. That is, the factor strength has to be stronger than the noise level.

In practice,  the factor strength $\sigma_{K^*+1}^2(M^*)$ is unknown while $N\vee J$ is observable. Therefore, we recommend to choose $u(n,N,J,K) = (N\vee J)h(n,N,J)$ for some slowly diverging factor $h(n,N,J)$, so that  %With this choice,
over-selection is avoided. Note that we require $h(n,N,J)$ to diverge slowly, so that under-selection is also avoided
for a wide range of factor strength levels. More specifically, we suggest to use $h(n, N, J)= \log\{n/(N\vee J)\}$ which becomes $\log(N\wedge J)$ when there is no missing data. Its consistency is established in Corollaries~\ref{coro:suggested-known} and \ref{coro:suggested} below.

%there is a close connection between the joint likelihood \eqref{eq:jointlik} and the marginal likelihood that treats the latent factors as random variables and integrate them out \citep{huber2004estimation}. Therefore, the theoretical development below may also be useful for the development of information criteria based on the marginal likelihood under a high-dimensional setting.

% and
%$u(n,N,J,K)$ is still close to the asymptotic lower bound $N\vee J$ so that it also avoids
%under a wide range of factor strength levels.

%it is easy to choose a penalty   $v(n,N,J,K)$ such that $N\vee J = o\big(u(n,N,J,K)\big)$ because $N$ and $J$ are directly observable. On the other hand,  depends on the true model parameters that are unknown.

 %Roughly speaking, the larger the value of $\sigma_{K^*+1}^2(M^*)$, the easier to detect the factors. On the other hand,
\end{remark}

%\clearpage
%
%
%We now discuss concrete choices of the penalty term under the setting of Theorem~\ref{thm:IC-simplified}. From \eqref{eq:simplified-ic-condition-u}, informally speaking,  $u(n,N,J,K)$ needs to be larger than $N\vee J$ and smaller than $\sigma_{K^*+1}^2(M^*)$, where $\sigma_{K^*+1}^2(M^*)$ can be viewed as the signal level of the data and $N\vee J$ measures the noise level. As our signal level $\sigma_{K^*+1}^2(M^*)$ is usually unknown in practice, we suggest to choose $u(n,N,J,K)$ to be close to the asymptotic lower bound $N\vee J$, so that it works for a wider range of signal levels.
%
%
%
%
%
%Specifically, consider the case when there is no missing data, i.e., $p_{min} = 1$. Then a sensible choice would be $u(n,N,J,K) = (N\vee J)\log (N\wedge J)$. This choice is implied by the penalty \eqref{eq:suggested} when there is no missing data. \yc{FURTHER DISCUSSIONS ON WHY THIS CHOICE.}
%We provide the following corollary of Theorem~\ref{thm:IC-simplified} to establish conditions under which this choice of $u(n,N,J,K)$ leads to consistent model selection. In the presence of missing data,
%the consistency of the suggested penalty \eqref{eq:suggested} will be established in Corollary~\ref{coro:suggested} as an implication of Theorem~\ref{thm:IC-general}.

\begin{corollary}\label{coro:suggested-known}
Assume that the asymptotic regime \eqref{eq:IC-simple-asym-reg} holds. Consider $v(n,N,J,K) = {K}({N \vee J})h(N,J)$ for some function $h$.
If $\lim_{N,J\to\infty}h(N,J)=\infty$ and $\lim_{N,J\to\infty}\{h(N,J)\}^{-1}(N\vee J)^{-1}\sigma_{K^*+1}^2(M^*)=\infty$,
then   $\lim_{N,J\to\infty}\Pr(\hat{K}=K^*)=1.$
Specifically, suppose that $p_{min} = 1$.
If $(N\vee J)\log(N\wedge J) = o\{\sigma_{K^*+1}^2(M^*)\}$ and we choose  $v(n,N,J,K) = {K}({N \vee J})\log(N\wedge J)$, then $\lim_{N,J\to\infty}\Pr(\hat{K}=K^*)=1$.
\end{corollary}

The next theorem extends Theorem~\ref{thm:IC-simplified} to a more general asymptotic setting.
\begin{theorem}\label{thm:IC-general}
Consider the following asymptotic regime as $N,J\to\infty$,
\begin{equation}\label{eq:IC-general-asym-reg}
   C=O(1) \text{ and } (N\wedge J)\log(N+J)=O\{n^*/(\log n^*)^2\} .
 \end{equation}
 Also, assume $p^{-2}_{\min}p_{\max}K^*(N\vee J)=o\{\sigma_{K^*+1}^2(M^*)\}$.
Suppose that there exists a possibly random sequence $\{\xi_{N,J}\}$ such that $\xi_{N,J}\to\infty$ in probability as $N,J\to\infty$,
  and
 with probability converging to one as $N,J\to\infty$, the following inequalities hold
 \begin{equation}\label{eq:cases}
 u(n,N,J,K)
   \begin{cases}
        \leq \xi_{N,J}^{-1}p_{\min}\sigma_{K+1}^2(M^*),  &\text{ if } 1\leq K\leq K^*,\\
        \geq \xi_{N,J}(K^*+1)(p_{\max}/p_{\min})(N\vee J),  &\text{ if } K=K^*+1,\\
        \geq \xi_{N,J}(p_{\max}/p_{\min})(N\vee J), &\text{ if } K^*+2\leq K\leq K_{\max},
   \end{cases}
 \end{equation}
where $K_{\max}\geq K^*$ denotes the largest number of factors considered in model selection and we allow $K_{\max}=\infty$.
Then, $\lim_{N,J\to\infty}\Pr(\hat{K}=K^*)=1.$
\end{theorem}
Theorem~\ref{thm:IC-general} relaxes the assumptions of Theorem~\ref{thm:IC-simplified} in several aspects. First, it is established under a more general asymptotic regime by allowing $K^*$ to diverge and $p_{\min}$ to decay to zero, as $N$ and $J$ grow. It also allows the missingness pattern to be very different from uniform sampling by allowing $p_{\max}/p_{\min}$ to grow. Second, $u(n,N,J,K)$ is allowed to be random  as long as \eqref{eq:cases} holds with high probability. In particular, the model selection consistency of the suggested penalty \eqref{eq:suggested} is established in Corollary~\ref{coro:suggested} below as an implication of Theorem~\ref{thm:IC-general}. Third, \eqref{eq:cases} provides a more specific requirement on $u(n,N,J,K)$. The second and third lines of \eqref{eq:cases} depend on the true number of factors $K^*$. In practice, we need to choose $u(n,N,J,K)$ in a way that does not depend on $K^*$. For example,  we may choose $u(n,N,J,K) = (K\wedge K_{\max})(p_{\max}/p_{\min})(N\vee J) h(n,N,J)$ for some sequence $h(n,N,J)$ that tends to infinity in probability as $N$ and $J$ diverge, so that the second and third lines of \eqref{eq:cases} are satisfied.

  % {\color{red}
  % \begin{enumerate}
  %   \item relax asymptotic regime in by removing $p_{\min}^{-1}=O(1), K^*=O(1)$. Handle more general situations. Allow many missing/non-uniform missing. Allow number of factors grows to infinity.
  %   \item allow $u$ to be random due to the missingness. For example, $u$ could be depending on $n$ and is thus random.
  %   \item If further assume $p_{\max}/p_{\min}=O(1)$ and $NJ=O(\sigma)$. Then .. give an example of $
  %   u=(K\wedge K_{\max})(N\vee J)\log(J) \sim n
  %   $
  %   \item The conditions depend on $K^*$ which is unknown in the practical problem. For practical purpose can take $v=K(K\wedge K_{\max})(N\vee J)\log(J)$ for example.
  % \end{enumerate}
  % }
\begin{corollary}\label{coro:suggested}
Assume the asymptotic regime \eqref{eq:IC-simple-asym-reg} holds and $N\vee J=o\{\sigma_{K^*+1}^2(M^*)\}$. Consider $ v(n,N,J,K) = {K}({N \vee J})h(n,N,J)
$.
If $h(n,N,J)\to\infty$ in probability as $N,J\to\infty$ and $\{h(n,N,J)\}^{-1}(N\vee J)^{-1}\sigma_{K^*+1}^2(M^*)\to\infty$ in probability as $N,J\to\infty$,
%as $N,J\to\infty$ such that $\lim_{N,J\to\infty}\Pr(h(n,N,J)\geq \xi_{N,J})=1$ and $\lim_{N,J\to\infty}\Pr\big(h(n,N,J)\leq \xi^{-1}_{N,J}\sigma_{K^*+1}^2(M^*)/(N\vee J)\big)=1$,
then   $\lim_{N,J\to\infty}\Pr(\hat{K}=K^*)=1.$ In particular, if we choose  $v(n,N,J,K) = {K}({N \vee J})\log\{n/(N\vee J)\} $ as suggested in \eqref{eq:suggested} and assume $(N\vee J)\log(N\wedge J) = o\{\sigma_{K^*+1}^2(M^*)\}$, then $\lim_{N,J\to\infty}\Pr(\hat{K}=K^*)=1$.
\end{corollary}
% {\color{red}I want to replace the condition of the corollary to $h(n,N,J)$ satisfies $h(n,N,J)\to\infty$ in probability as $N,J\to\infty$ and $(h(n,N,J))^{-1}(N\vee J)^{-1}\sigma_{K^*+1}^2(M^*)\to\infty$ in probability as $N,J\to\infty$, which is equivalent  if the following statement is true:
% $$
% X_n\to \infty\text{ in probability } \Leftrightarrow \exists \text{ deterministic } M_n\to\infty s.t. \lim_{n}\Pr(X_n\geq M_n) =1.
% $$
% It seems a correct statement but I don't know how to show it..Keep the current form of corollary?
% }

% {\color{red}
% Should we give a lower bound to claim that the proposed one is almost smallest penalty for guaranteeing consistency? We can give a remark at least.

% }
\begin{remark}\label{rmk:dispersion}
In Theorems~\ref{thm:IC-simplified} and \ref{thm:IC-general}, the dispersion parameter $\phi$ is assumed known. When $\phi$ is unknown, we may first fit the largest model with $K_{\max}$ factors to  obtain an estimate $\hat{\phi}$, and then select the number of factors using JIC with $\phi$ replaced by $\hat{\phi}$. Similar model selection consistency results would still hold. We note that the use of the plug-in estimator for dispersion parameter is common in constructing information criteria for linear models and linear factor models \citep{bai2002determining}.

\begin{remark}\label{rmk:connection}
Several information criteria have been proposed for linear factor models under high-dimensional regimes. In particular,
\cite{bai2002determining} consider a setting that the observed data matrix can be decomposed as the sum of a low rank matrix and a mean-zero error matrix, and propose information criteria to select the rank of the low-rank matrix. Their setting is very similar to the case when the exponential family distribution in \eqref{eq:models} is chosen to be a Gaussian distribution and there is no missing data, except that \cite{bai2002determining} do not require the Gaussian assumption. In fact, under the Gaussian linear factor model and when the dispersion parameter $\phi = 1$, the suggested JIC with penalty term $v(n,N,J,K) = {K}({N \vee J})\log(N\wedge J)$ is
asymptotically equivalent to the  $PC_{p1}$ through  $PC_{p3}$ criteria proposed in \cite{bai2002determining} and in particular takes the same form as
the
$PC_{p3}$ criterion.
%the same as , and is also
 \cite{bai2018consistency} consider the spike covariance structure model \citep{johnstone2001distribution} and develop information criteria for choosing the number of dominant eigenvalues which corresponds to the number of factors when regarding the spike covariance structure model as a linear factor model. By random matrix theory, they establish consistency results  when the sample size and the number of manifest variables grow to infinity at the same speed and there is no missing data.

As mentioned in Section~\ref{sec:intro}, nonlinear factor models are more suitable for multivariate data that involve categorical or count variables. Specifically, under model \eqref{eq:models}, the expected data matrix is $\{b'(m_{ij})\}_{N\times J}$. Although $M = (m_{ij})_{N\times J}$ is a low-rank matrix, $\{b'(m_{ij})\}_{N\times J}$ is no longer a low-rank matrix when $b'$ is a nonlinear transformation. Consequently, methods developed for the linear factor model do not work well when data follow a nonlinear factor model. The presence of massive missing data further complicates the problem.

\end{remark}

Finally, we point out that the theoretical results established above may also be useful for developing information criteria based on the marginal likelihood. The marginal likelihood, which is widely used for estimating latent variable models, treats the latent factors as random variables and integrates them out.
When both $N$ and $J$ are large, by applying the Laplace approximation \citep{huber2004estimation}, the marginal likelihood can be approximated by a joint likelihood plus some remainder terms. The development above can be used to analyze this joint likelihood term.
%behaves like the joint likelihood.
Further discussions are given in Appendix~\ref{app:marginal} of the online supplement.

%Therefore, under suitable conditions, the same penalty term should also work if we replace the joint maximum likelihood by the marginal maximum likelihood.

%  {\color{red} Point out that the dispersion parameter can be estimator at the same error rate. We will then use a version of JIC with the estimated dispersion parameter added to the penalty term. }
\end{remark}
\section{Numerical Experiments}\label{sec:examp}
\subsection{Simulation}

We use a simulation study to evaluate the model estimation and the selection of factors with the proposed JIC with $v(n,N,J,K) = {K}({N \vee J})\log\{n/(N\vee J)\}$. Due to the space constraint, we only present some of the results under the logistic factor model for binary data. Additional results from this study and results from other simulation studies can be found in Appendix \ref{app:addsim} of the online supplement.

In particular, we consider eight combinations of $N$ and $J$, given by $J = 100, 200, 300, 400$, $N = J$ and $N = 5J$. We consider three settings for missing data, including (M1) no missing data, (M2) uniformly missing, with missingness probability $p_{ij} = 0.5$ for all $i$ and $j$, and (M3) non-uniformly missing, with missingness probability
$p_{ij} =  {\exp(f_{i1}^*)}/{\{1+\exp(f_{i1}^*)\}}$ that depends on the value of the first factor. The true number of factors is set to $K^* = 3$.
The model parameters are generated as follows. First, the true parameters $d_j^*$, $a_{j1}^*$, ..., $a_{j3}^*$ are generated by sampling independently from the uniform distribution over the interval $[-2,2]$. Second, the true factor values are generated under two settings. Under the first setting (S1), all three factors $f_{i1}^*$, ..., $f_{i3}^*$ are
generated by sampling independently from the uniform distribution over the interval $[-2,2]$, so that all the factors have essentially the same strength. Under the second setting (S2), the first two factors $f_{i1}^*$ and $f_{i2}^*$ are
generated the same as those under the first setting, while the last factor $f_{i3}^*$ is sampled from the uniform distribution over the interval $[-0.8, 0.8]$. Under the second setting, the last factor tends to be weaker than the rest and thus is more difficult to detect. We use the proposed JIC to select $K$ from the candidate set $\{1,2,3,4, 5\}$ and the constraint constant $C$ in \eqref{eq:jml} is set to be 5. The true model parameters satisfy this constraint. All the combinations of the above settings lead to 48 different simulation settings. For each setting, we run 100 independent replications. The computation is done using the R package \texttt{mirtjml}  \citep{zhang2020mirtjml}.

We first examine the results on parameter estimation. The loss
$\max_{3\leq K\leq 5} \{(NJ)^{-1/2} \|\hat{M}^{(K)}-M^* \|_F \}$ under different settings is shown in Figure~\ref{fig:1}. As we can see, under each setting for factor strength and missing data, the loss decays towards zero, as both $N$ and $J$ grow. Given the same $N$ and $J$, the estimation tends to be more accurate when there is no missing data. In addition, the estimation tends to be more accurate under setting M2 where the data entries are uniformly missing than that under M3 where the missingness depends on the latent factors. We further examine the selection of factors. Table~\ref{tab:tab1} presents the frequency that the number of factors is under- and over-selected among the 100 independent replications for all the 48 settings. As we can see, the proposed JIC becomes more accurate as $N$ and $J$ grow.
Under the settings when $N = J$, no under-selection is observed, but the proposed JIC is  likely to over-select when  $J$ is relatively small. Under the settings when $N = 5J$, no over-selection is observed, but  under-selection is observed when one factor is relatively weaker than the others and $J$ is relatively small. We point out that determining the number of factors is a challenging task under our settings when  $J$ is relatively small. To illustrate,  Figure~\ref{fig:signal} shows the box plots of $2(\hat l_3 - \hat l_2)$ and $2(\hat l_4 - \hat l_3)$ under settings when $J = 100$. For most of these settings, $2(\hat l_3 - \hat l_2)$ is not substantially larger than $2(\hat l_4 - \hat l_3)$, while our asymptotic theory requires the former to be of a higher order.
%{\color{red} (XL: I don't understand: why $N=J$ and $N=5J$ have very different patterns?)}

From the results in Table~\ref{tab:tab1}, we see that for relatively small values of $N$ and $J$,
the proposed information criterion tends to over-penalize when
$N = 5J$ and under-penalize when $N=J$. We explain this phenomenon.
Our choice of $v(n,N,J,K)$ is derived from the error bound \eqref{eq:error-bound-simple} in Theorem~\ref{thm:error-bound-simple}. Although this error bound is rate optimal as implied by Proposition~\ref{prop:lower-bound}, it does not take into account the relationship between $N$ and $J$. For example, consider two settings that both have no missing data and the same $J$, but one with $N = J$ and the other with $N=5J$. By Theorem~\ref{thm:error-bound-simple}, the two settings have exactly the same upper bound $\kappa(K_{max}/J)^{1/2}$. However, as we can see from Fig. \ref{fig:1}, the error tends to be larger
under the setting when $N=J$ than that when $N =5J$. Consequently, with the JIC derived from the same upper bound, it is more likely to over-select when $N=J$ and to under-select when $N=5J$. To improve the current information criterion, a refined error bound is needed, according to which we can choose a $v(n,N,J,K)$ that better adapts to the relationship between $N$ and $J$. This is a challenging problem and we leave it for future investigation.

%the two panels of Figure~\ref{} show the box plots of 2()

%We provide some explanations on the
%different behaviors under different relationships between $N$ and $J$.  As discussed in Remark~\ref{rmk:5}, the penalty $v(n,N,J,K)$ is chosen, so that

%\clearpage

%Figure~\ref{fig:2} below shows the distributions of $2(\hat l_5 -\hat l_4)$ and $2 (\hat l_6 -\hat l_5)$ under different simulations and compares them with $u(n,N,J,K)$. Note that in order to correctly select the number of factors, we need $u(n,N,J,K)$ to be smaller than $2(\hat l_5 -\hat l_4)$ while larger than $2 (\hat l_6 -\hat l_5)$. The relative sizes of $2(\hat l_5 -\hat l_4)$ and $2 (\hat l_6 -\hat l_5)$ determine the difficulty of correctly selecting the number of factors. \yc{Further discussions on the plots. }  \yc{Discuss the results in the table. Further describe the under-selection problem and explain the reason.}

\begin{figure}[t]
  \centering
  \begin{subfigure}[b]{0.49\textwidth}
         \centering
         \includegraphics[width=\textwidth]{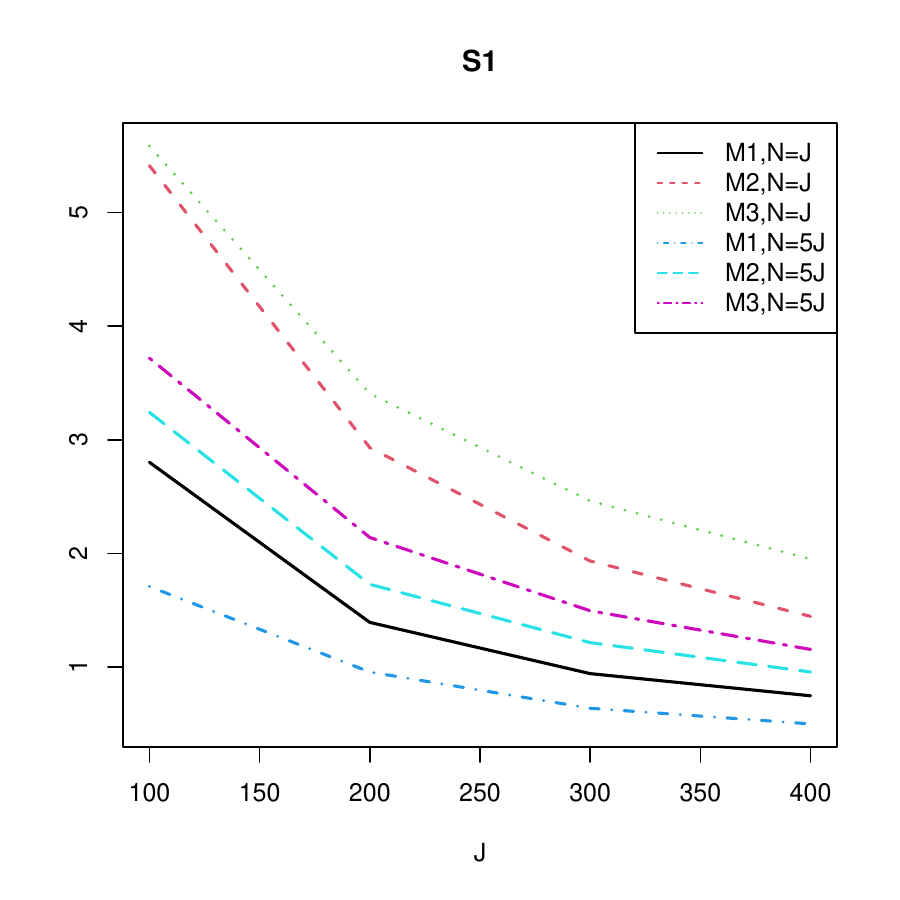}
         \label{fig:y equals x}
     \end{subfigure}
     \begin{subfigure}[b]{0.49\textwidth}
         \centering
         \includegraphics[width=\textwidth]{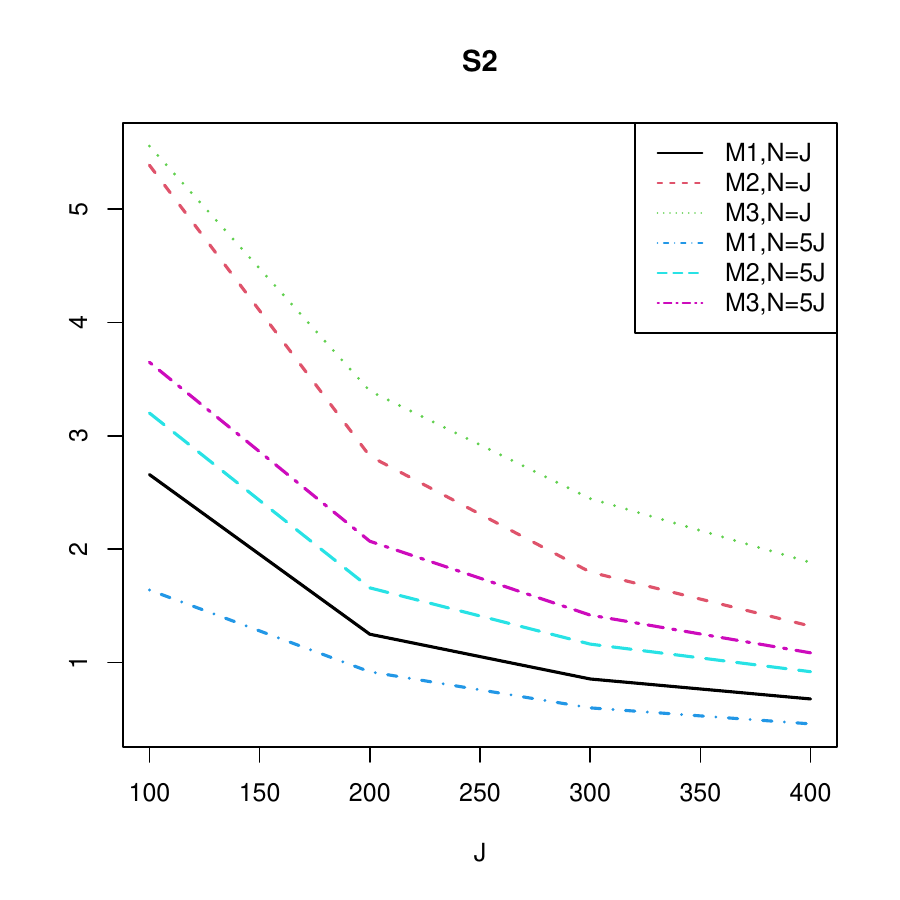}
         \label{fig:three sin x}
     \end{subfigure}
  \caption{The loss $\max_{3\leq K\leq 5} \{(NJ)^{-1/2} \|\hat{M}^{(K)}-M^* \|_F \}$ for the recovery of the low-rank matrix $M^*$, where each point is the mean loss calculated by averaging over 100 independent replications. Panels (a) and (b) show the results under the two different factor strength settings, S1 and S2, respectively.}\label{fig:1}
\end{figure}

% twice of the difference between the maximum likelihood when $K = 5$ and that when $K = 4$, and twice of the difference between the maximum likelihood when $K = 6$ and that when $K = 5$.

\begin{table}
\caption{The number of times that the true number of factors is under- or over-selected selected among 100 independent replications under each of the 48 simulation settings.}
\centering
\begin{tabular}{r|ccc|ccc|ccc|cccccccccccc}
  \hline
  & \multicolumn{6}{c|}{$N = J$} & \multicolumn{6}{c}{$N = 5J$} \\
  \hline
 & \multicolumn{3}{c|}{S1} & \multicolumn{3}{c|}{S2} & \multicolumn{3}{c|}{S1} & \multicolumn{3}{c}{S2}\\
  \hline
 Under-selection &M1 & M2 & M3 &M1 & M2 & M3 &M1 & M2 & M3 &M1 & M2 & M3 \\
  \hline
  $J = 100$&0&0&0 &0&0&0& 0&0&0& 10& 98& 97 \\
  $J = 200$&0&0&0 &0 & 0&  0 & 0&0&0&0 & 3&  4 \\
  $J = 300$&0&0&0 &0&0&0& 0&0&0&0&0&0 \\
  $J = 400$&0&0&0 &0&0&0& 0&0&0&0&0&0 \\
  \hline
   Over-selection&  \\
   \hline
  $J = 100$&47 &100& 100 & 53& 100& 100&0  & 0&   0 &0&   0&   0 \\
  $J = 200$&0&94&90      &0 &100&  95&0  & 0&   0 &0 &  0& 0\\
  $J = 300$&0& 0&0      &0&0&0&0  & 0&   0 &0&0&0 \\
  $J = 400$&0& 0&0      &0&0&0&0  & 0&   0 &0&0&0 \\
  \hline
\end{tabular}
\label{tab:tab1}
\end{table}

\begin{figure}
  \centering
  \includegraphics[width=\textwidth]{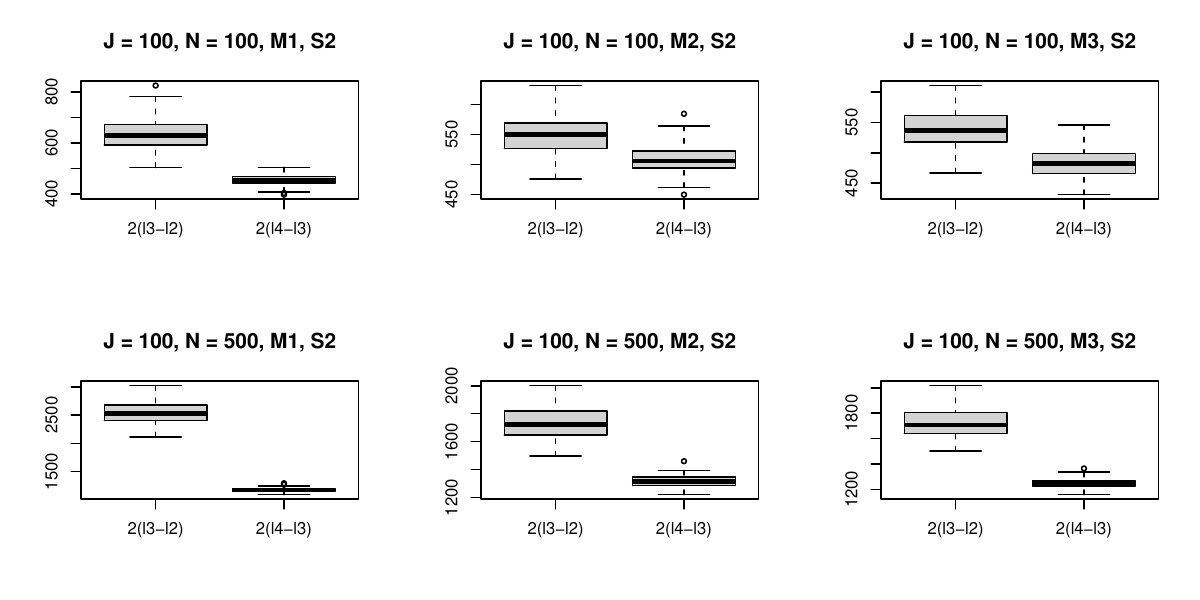}
  \caption{The box plots of $2(\hat l_3 - \hat l_2)$ and $2(\hat l_4 - \hat l_3)$ when $J = 100$ and the factor strength setting is S2.}\label{fig:signal}
\end{figure}

\subsection{Application to Eysenck's Personality Questionnaire}
We apply the proposed JIC to a dataset based on the revised Eysenck's personality questionnaire  \citep{eysenck1985revised}, a personality inventory that has been widely used in clinics and research. This questionnaire is designed to measure three personality traits, extraversion, neuroticism, and psychoticism.
We refer the readers to \cite{eysenck1985revised} for the  characteristics of these personality traits.
{The factor structure of this personality inventory remains of interest in psychology, due to its importance in the literature of human personality and wide use in several studies worldwide \citep{barrett1998eysenck,chapman2013hierarchical,heym2013p}. In particular, it has been found that the dependence between items measuring the psychoticism trait tends to be lower than that between items measuring the other two traits. Based on this observation, some researchers suggested that  psychoticism may consist of multiple dimensions \citep{caruso2001reliability}. We use the proposed JIC to investigate the factor structure of the inventory. }

Specifically, we analyze all the items from the questionnaire, except for the lie scale items that are used to
to guard against various concerns about response style. There are 79 items in total, each with ``Yes" and ``No" response options. An example item is ``Do you often need understanding friends to cheer you up?".
Among the 79 items, 32, 23, and 24 items are designed to measure psychoticism, extraversion, and neuroticism, respectively. {For each participant, a total score can be computed based on each of the three item sets. This total score is often used to measure the corresponding personality trait.}
%The determination of the number of latent factors underlying this personality questionnaire is important in personality psychology. Specifically, the number of factors underlying the revised Eysenck's personality questionnaire  has been studied using data from many countries \citep{barrett1998eysenck}.
Here, we analyze a female UK normative sample dataset \citep{eysenck1985revised}, for which the sample size is 824 and there are no missing values. The dataset has been analyzed in \cite{chen2019joint} using the same  model given in Example 1 above. Using a cross-validation approach, \cite{chen2019joint} find three factors. We now explore the dimensionality of the data using the proposed JIC. Specifically, we consider possible choices of $K = 1, 2, 3, 4, $ and 5.  Following the previous discussion, the penalty term in the JIC is set to ${K}({N \vee J})\log\{n /(N \vee J)\}$, where $n = NJ$, $N = 824$, and $J = 79$.

{The results are given in Tables~\ref{tab:JICepq} and \ref{tab:epqCOR}. Specifically, the three-factor model achieves the minimum JIC value among the five candidate choices of $K$, suggesting a three-factor structure for the inventory.
%suggesting there are three factors underlying this personality questionnaire.
We investigate the three-factor model using the oblimin method, one of the most popular oblique rotation methods \citep{browne2001overview}, to obtain an relatively simple loading structure. Table~\ref{tab:epqCOR} shows Kendall's tau rank correlation between participants' estimated factor scores under the oblimin rotation and the total scores for the three personality traits given by the design. According to the correlations, the extracted factors tend to correspond to the extraversion, psychoticism,  and neuroticism traits, respectively.
Additional results can be found in Appendix~\ref{app:real} of the online supplement, including the estimated parameters for the fitted models and a comparison with marginal-likelihood-based inference.
}

\begin{table}
  \caption{The rows named ``deviance", ``penalty", and ``JIC" show the values of $-2 \hat l_K$, $v(n,N,J,K)$, and JIC, respectively, for models with different values of $K$.}
  \centering
  \begin{tabular}{lccccc}
    \hline
    $K$ & 1&2&3&4&5 \\
    \hline
    Deviance    & 63263 & 57683 &53883 & 51225 & 48812 \\
    Penalty &        3600  &  7201  &10801 &14402  &18002            \\
    JIC &   66864 &64884     & \textbf{64684} &65627  & 66814        \\
    \hline
  \end{tabular}
\label{tab:JICepq}
\end{table}

%\begin{figure}
%  \centering
%  \includegraphics[scale = 0.5]{real.pdf}
%  \caption{Panels (a) and (b) show the values of deviance and JIC, respectively. }\label{fig:real}
%\end{figure}

\begin{table}
  \caption{Kendall's tau rank correlation between participants' estimated factor scores under the oblimin rotation and the total scores for the three personality traits.}
  \centering
  \begin{tabular}{lrrrr}
    \hline
     & F1&F2&F3  \\
    \hline
P score&  0.08& \textbf{0.78}& -0.05\\
E score&  \textbf{0.86}& 0.00& -0.12\\
N score& -0.08& 0.08&  \textbf{0.88}\\
    \hline
  \end{tabular}
\label{tab:epqCOR}
\end{table}

%On the other hand, it is also worth noting that
%the JIC value of the two-factor model is quite close to that of the three-factor model, suggesting that the signal in the third factor may not be very strong.

\section{Further Discussion}
As shown in Section~\ref{section:theory}, there is a wide range of penalties for guaranteeing the selection consistency of JIC. Among these choices,
$v(n,N,J,K) = {K}({N \vee J})\log\{n/(N\vee J)\}$ is close to the lower bound.
This penalty is suggested when the signal strength of factors is unknown, to detect factors of a wide range of strengths. The performance and applicability of this information criterion are demonstrated by simulation studies and real data analysis.
%The performance and applicability of this information criterion
%According to simulation study and real data analysis in Section~\ref{sec:examp}, this penalty choice performs well for both the logistic and Poisson factor models.
If one is only interested in detecting strong factors, then a larger penalty may be chosen based on prior information about the signal strength of the factors.

When our model  \eqref{eq:models} takes the form of a Gaussian density and there is no missing data, then the proposed JIC and its theory
are consistent with the results of \cite{bai2002determining} for high-dimensional linear factor models. In this sense, the current work substantially extends the work of
\cite{bai2002determining} by considering non-linear  factor models and allowing a general setting for missing values.
%In particular, our missing data setting allows the data entries to be non-uniformly missing and the expected proportion of missing entries to grow to one.
 %are very similar to those given in \cite{bai2002determining}.
Although we focus on generalized latent factor models with an exponential-family link function, the proposed JIC is applicable to other models, for example,
a probit factor model for binary data that replaces
the logistic link by a probit link in Example 1. The consistency results are likely to hold under similar conditions, for a wider range of models.
This extension is left for future investigation.

\section*{Acknowledgement}

We are grateful to the editor, the associate editor and three anonymous referees for their careful review and valuable comments. Xiaoou Li's research is partially supported by the NSF grant DMS-1712657.

%additional simulation results and
%
%
%\begin{figure}[t]
%% The arguments in the next line are {height}{optional width}{used only by OUP for typesetting} for figure empty box eg
%%\figurebox{20pc}{25pc}{}
%%if actual size of graphics need plese see below command
%%\figurebox{}{}{}[fig1]
%%need to reducing the figure size use below command
%%\figuresize{.8}%
%\figuresize{.58}
%\figurebox{20pc}{25pc}{}[fig1]
%% note that files may not be rotated
%\caption{A graph showing the truth (dot-dash), an estimate (dashes), another estimate (solid), and 95\% pointwise confidence limits (small dashes).}
%\label{fig1}
%%\vspace*{-30pt}
%\end{figure}
%
%
%\begin{table}
%\def~{\hphantom{0}}
%\tbl{Perceptions about racial groups in the U.S. population}{%
%\begin{tabular}{lcccc}
%%\\
% \\
%& 2000 Census& \multicolumn{3}{c}{Mean percent estimated for U.S.} \\
%& percent of  & \multicolumn{3}{c}{population } \\
%& U.S. population & White Rs & Black Rs & Hispanic Rs \\[5pt]
%White & 75 & 59 & 56 & 60 \\
%Black & 12 & 30 & 38 & 40 \\
%Asian & ~4 & 16 & 21 & 30 \\
%American Indian & ~1 & 13 & 17 & 23 \\
%More than two races & ~2 & 41 & 48 & 50 \\
%Hispanic & 13 & 23 & 27 & 42
%\end{tabular}}
%\label{tablelabel}
%\begin{tabnote}
%U.S., United States of America; R, respondent.
%\end{tabnote}
%\end{table}

%\appendix

%\appendixone

\clearpage
\appendix

\noindent
{\Large Appendix}

\section{Proof of Theoretical Results}
\subsection{Proof of Theorems~\ref{thm:error-bound-simple} and \ref{thm:finite-bound}}
We will  present the proof of Theorem~\ref{thm:finite-bound} first and then that of Theorem~\ref{thm:error-bound-simple}, because the former is more general than the latter.
The proof of Theorem~\ref{thm:finite-bound} is based on the following two lemmas, whose proof will be provided later in the supplementary material.

%Appendix~\ref{sec:proof-lemmas}.

% We first provide three useful lemmas first. The proof of these lemmas are given later this section.
Let $G_i=(1,F_i^T)^T$ and  $B_j= (d_j,A_j^T)^T$, then $m_{ij}=G_i^TB_j$. Define $\mM_r=\{\MM=(m_{ij})_{1\leq i\leq N,1\leq j\leq J}: m_{ij} =  G_i^TB_j : G_i,B_j\in  \mathbb{R}^r, \|G_i\|\leq C, \|B_j\|\leq C \text{ for all }1\leq i\leq N,1\leq j\leq N\}$, then  $M^*\in \mM_{K^*+1}$. Let $r^*=K^*+1$ under Assumption~\ref{assump:true}. Also, let $l(M,Y,\Omega)$ denote the log-likelihood function where $\Omega=(\omega_{ij})_{1\leq i\leq N,1\leq j\leq J}$.
\begin{lemma}\label{lemma:correct-specify-upper}
For all $M\in\mM_r$,
 \begin{equation}
 \begin{split}
     &\phi\{l(M,Y,\Omega)-l(M^*,Y,\Omega)\}\\\leq &
(r+r^*)^{1/2}\big\{\|\ZZ\circ\Omega\|_2+ 2\delta_{C^2}C^2\|Q\|_2\big\}\|\MM-\MM^*\|_F- \delta_{C^2}p_{\min}\|\MM-\MM^*\|_F^2
 \end{split}
 %   (K+K^*)^{1/2}\phi\|Z\circ\Omega\|_2\|M-M^*\|_F,
 \end{equation}
 where $Z=(z_{ij})_{1\leq i\leq N,1\leq j\leq J}$, $z_{ij}=y_{ij}-b'(m^*_{ij})$, and
 `$\circ$' denotes the matrix Hadamard product.
\end{lemma}
\begin{lemma}\label{lemma:spectral-pattern}
  There is a universal constant $c>0$ such that
  \begin{equation}
    \Pr\Big(\|\Omega-P\|_2 \geq 4 (\max_{i} n^*_{i\cdot})^{1/2} \vee (\max_{j} n^*_{\cdot j})^{1/2}+c \log^{1/2}(N+J) \Big)\leq (N+J)^{-1}.
  \end{equation}
\end{lemma}
\begin{lemma}\label{lemma:spectral-subexpnential-with-missing}
    Let $V=(v_{ij})_{1\leq i\leq N, 1\leq j\leq J}$ be a random matrix with independent and centered entries. In addition, assume $v_{ij}$s are sub-exponential random variables with parameters $\nu,\alpha>0$. That is, $E(e^{\lambda v_{ij}})\leq e^{\lambda^2\nu^2/2}$ for all $|\lambda|<1/\alpha$.
  Then, there exists a universal constant $c>0$ such that with probability at least $1-(N+J)^{-1}-(n^*)^{-1}$,
  \begin{equation}\label{eq:spectral-subexp}
   \|V\circ \Omega\|_2\leq 4\max_{ij}\{E(z_{ij}^2)\}^{1/2}(\max_{i} n_{i\cdot}^*)^{1/2}\vee (\max_{j} n_{\cdot j}^*)^{1/2}  + c (\alpha\vee \nu)\log n^* \log^{1/2}(N+J)
   \end{equation}
   for all $N\geq 1, J\geq 1$, and $n^*\geq 6$.
   In particular, under Assumptions 1 and 2, $z_{ij}=y_{ij}-b'(m_{ij}^*)$ is sub-exponential with parameters $\nu^2=\phi\kappa_{2C^2}=\phi\sup_{|x|\leq 2C^2}b''(x)$  and $\alpha=\phi/C^2$, and there is a universal constant $c>0$ such that with probability at least $1-(N+J)^{-1}-(n^*)^{-1}$,
   \begin{equation}
     \|Z\circ \Omega\|_2\leq  4(\phi\kappa_{2C^2})^{1/2}(\max_{i} n_{i\cdot}^*)^{1/2}\vee (\max_{j} n_{\cdot j}^*)^{1/2}  + c \{(\phi/C^2)\vee (\phi\kappa_{2C^2})^{1/2}\}\log n^* \log^{1/2}(N+J)
   \end{equation}
      for all $N\geq 1, J\geq 1$, and $n^*\geq 6$.
\end{lemma}
\begin{remark}
The constant $4$ in the first term of the right-hand side of \eqref{eq:spectral-subexp} can be improved to $2\sqrt{2}+\epsilon$ for any $\epsilon>0$ with the constant $c$ replaced by an $\epsilon$-dependent constant $c_{\epsilon}$. The logarithm term  can be improved if $Z$ is further assumed sub-Gaussian or bounded.  We keep the current form which is sharp enough for   our   problem.
\end{remark}
\begin{proof}[of Theorem~\ref{thm:finite-bound}]
By the definition of $\hat{M}^{(K)}$ and $K\geq K^*$, we have $\phi\{l(\hat{M}^{(K)},Y,\Omega)-l(M^*,Y,\Omega)\}\geq 0$. Apply Lemma~\ref{lemma:correct-specify-upper} with $M=\hat{M}^{(K)}$, $r=K+1$ and combine it with $\phi\{l(\hat{M}^{(K)},Y,\Omega)-l(M^*,Y,\Omega)\}\geq 0$. We obtain that for every $K\geq K^*$,
\begin{equation}\label{eq:error-bound-mle0}
  \|\hat{M}^{(K)}-M^*\|_F\leq p_{\min}^{-1}
(K+K^*+2)^{1/2}\big\{\delta_{C^2}^{-1}\|\ZZ\circ\Omega\|_2+ 2C^2\|Q\|_2\big\}.
\end{equation}
Thus,
\begin{equation}\label{eq:error-bound-mle}
\begin{split}
    \max_{K^*\leq K\leq K_{\max}}\big(\|\hat{M}^{(K)}-M^*\|_F\big)\leq 2p_{\min}^{-1}
K_{\max}^{1/2}\big\{\delta_{C^2}^{-1}\|\ZZ\circ\Omega\|_2+ 2C^2\|Q\|_2\big\},
\end{split}
\end{equation}
where we used the fact that $K+K^*+2\leq 2(K_{\max}+1)\leq 4 K_{\max}$ for $K_{\max}\geq 1$.
Apply Lemma~\ref{lemma:spectral-pattern} and Lemma~\ref{lemma:spectral-subexpnential-with-missing}  to obtain an upper bound of the right-hand side of the above inequality and simplify it. We arrive at
\begin{equation}
\begin{split}
  & \max_{K^*\leq K\leq K_{\max}}\big(\|\hat{M}^{(K)}-M^*\|_F\big)\\
   \leq & 2p_{\min}^{-1}
(K_{\max})^{1/2}\big[\{4\delta_{C^2}^{-1}(\phi\kappa_{2C^2})^{1/2}+8C^2\}(\max_{i} n_{i\cdot}^*)^{1/2}\vee (\max_{j} n_{\cdot j}^*)^{1/2} \\
&~~~~~~~~~~~~~~~~~~~~~~~~~ + c \{(\phi/C^2)\vee (\phi\kappa_{2C^2})^{1/2}\log n^*+2C^2\} \log^{1/2}(N+J)\big]\\
=&   p_{\min}^{-1}
(K_{\max})^{1/2}\{ \kappa_{1,b,C,\phi}(\max_{i} n_{i\cdot}^*)^{1/2}\vee (\max_{j} n_{\cdot j}^*)^{1/2} + 2c(\kappa_{2,b,C,\phi}\log n^*+2C^2)\log^{1/2}(N+J) \}
\end{split}
\end{equation}
where we recall that $\kappa_{1,b,C,\phi}=8\delta_{C^2}^{-1}(\phi\kappa_{2C^2})^{1/2}+16C^2$ and $\kappa_{2,bC,\phi}=(\phi/C^2)\vee (\phi\kappa_{2C^2})^{1/2}$. This completes our proof.
\end{proof}

\begin{proof}[of Theorem~\ref{thm:error-bound-simple}]
  Note that $\max_{i}n_{i\cdot}^* \leq p_{\max} J$ and $\max_{j}n_{\cdot j}^* \leq p_{\max} N$. Thus, \eqref{eq:finite-bound-thm} is simplified to
  \begin{equation}
    \max_{K^*\leq K\leq K_{\max}}\big(\|\hat{M}^{(K)}-M^*\|_F\big)
    \leq \kappa (K_{\max})^{1/2}\{ p_{\min}^{-1/2}(N\vee J)^{1/2} + p_{\min}^{-1}\log n^*\log^{1/2}(N+J) \}
  \end{equation}
  for some $\kappa$ depending on $C,b,\phi$ and $p_{\max}/p_{\min}$. Because $p_{\min}= (p_{\min}/p_{\max}) p_{\max}\geq (p_{\min}/p_{\max})n^*/(NJ)$, the above inequality implies
  \begin{equation}
  \begin{split}
       & \max_{K^*\leq K\leq K_{\max}}\big(\|\hat{M}^{(K)}-M^*\|_F\big)\\
    \leq & \kappa (K_{\max})^{1/2}\{ (n^*/(NJ))^{-1/2}(N\vee J)^{1/2} + (n^*/(NJ))^{-1}\log(n^*)\log^{1/2}(N+J) \}
  \end{split}
  \end{equation}
  with a possibly different $\kappa$ that also depends on $C$, $b$, and $\phi$. Multiplying both sides by $(NJ)^{-1/2}$ and simplifying it, we arrive at
  \begin{equation}
  \begin{split}
      &\max_{K^*\leq K\leq K_{\max}}\big\{(NJ)^{-1/2}\|\hat{M}^{(K)}-M^*\|_F\big\}\\
      \leq & \kappa K^{1/2}_{\max}\Big[
          \big\{(N\vee J)/n^*\big\}^{1/2} + \{(NJ)^{1/2}\log^{1/2}(N+J) \} (n^*)^{-1}\log n^*
          \Big].
  \end{split}
  \end{equation}
  Note that for $n^*/(\log n^*)^2\geq (N\wedge J)\log(N+J)$,  $\big\{(N\vee J)/n^*\big\}^{1/2}\geq \{(NJ)^{1/2}\log^{1/2}(N+J) \} (n^*)^{-1}\log n^*$, and the above inequality is simplified as
\begin{equation}
  \begin{split}
      \max_{K^*\leq K\leq K_{\max}}\big\{(NJ)^{-1/2}\|\hat{M}^{(K)}-M^*\|_F\big\}
      \leq  2\kappa           \big\{K_{\max}
(N\vee J)/n^*\big\}^{1/2}.
  \end{split}
  \end{equation}
  This completes the proof.
\end{proof}
\subsection{Proof of Theorems~\ref{thm:IC-simplified} and \ref{thm:IC-general} and Corollary~\ref{coro:suggested}}
The proofs of Theorems~\ref{thm:IC-simplified} and  \ref{thm:IC-general} are based on the following three supporting lemmas, whose proofs are given in the supplementary material. We start by recalling $u(n,N,J,K)=v(n,N,J,K)-v(n,N,J,K-1)$ and defining $R=4(p_{\min}\delta_{C^2})^{-1}\big\{\|\ZZ\circ\Omega\|_2+ 2\delta_{C^2}C^2\|Q\|_2\big\}^2$.
\begin{lemma}\label{lemma:overselect}
 If $u(\cdot)$ satisfies
  \begin{equation}
   \lim_{N,J\to\infty} \Pr\Big(u(n,N,J,K^*+1)>2\phi^{-1}(K^*+1)R\Big) =1
  \end{equation}
  and
  \begin{equation}
    \lim_{N,J\to\infty} \Pr\Big(\inf_{K^*+2\leq K\leq K_{\max}}u(n,N,J,K)>2\phi^{-1}R\Big) = 1,
  \end{equation}
  then
  \begin{equation}
    \lim_{N,J\to\infty}\Pr(\hat{K}>K^*) =0,
  \end{equation}
  for $K_{\max}\geq K^*\geq 1$.
\end{lemma}
\begin{lemma}\label{lemma:underselect}
 If
\begin{equation}
  \lim_{N,J\to\infty}\Pr\Big(4(\delta_{C^2}p_{\min})^{-1}K^*R \leq \sigma^2_{K^*+1}(M^*)\Big)=1,
\end{equation}
 % $4(\delta_{C^2}p_{\min})^{-2}(2K^*)\big\{\|\ZZ\circ\Omega\|_2+ 2\delta_{C^2}C^2\|Q\|_2\big\}^2\leq  \sigma^2_{K^*}(M^*)$ {\color{red} ($2(\delta_{C^2}p_{\min})^{-1}R(2K^*)\leq \sigma^2_{K^*}(M^*)$ )},
 and $u(\cdot)$ satisfies

\begin{equation}
 \lim_{N,J\to\infty}  \Pr\Big(
   u(n,N,J,K)< \phi^{-1}\delta_{C^2}p_{\min}\sigma^2_{K+1}(M^*) \text{ for all } 1\leq K\leq K^*
  \Big)=1
 \end{equation}
 then $\lim_{N,J\to\infty}\Pr(\hat{K}<K^*)=0$ for $K^*\geq 1$.
\end{lemma}
\begin{lemma}\label{lemma:R-bound}
  %Recall $R=4(p_{\min}\delta_{C^2})^{-1}\big\{\|\ZZ\circ\Omega\|_2+ 2\delta_{C^2}C^2\|Q\|_2\big\}^2$.
  Under the asymptotic regime \eqref{eq:IC-general-asym-reg}, $R=O_p(p_{\max}/p_{\min} (N\vee J))$.
\end{lemma}
% {\color{red} Do we need to define $O_p()$ and $o_p()$, or it is well-understood by statistics audience?}
In the rest of the section, we provide the proof of Theorem~\ref{thm:IC-general} first and then the proof of Theorem~\ref{thm:IC-simplified} because the former is more general than the latter.

\begin{proof}[of Theorem~\ref{thm:IC-general}]
We will verify that conditions of Theorem~\ref{thm:IC-general} ensure conditions in Lemma~\ref{lemma:overselect} and Lemma~\ref{lemma:underselect}.
We start with verifying conditions in Lemma~\ref{lemma:overselect}.
According to the second line of \eqref{eq:cases},
 \begin{equation}
 \begin{split}
      &\lim_{N,J\to\infty} \Pr\Big(u(n,N,J,K^*+1)>2 \phi^{-1}(K^*+1)R\Big)\\
       \geq & \liminf_{N,J\to\infty} \Pr\Big(\xi_{N,J}(K^*+1)(p_{\max}/p_{\min})(N\vee J)> 2 \phi^{-1}(K^*+1)R\Big)\\
       \geq & \liminf_{N,J\to\infty} \Pr\Big(\xi_{N,J}(p_{\max}/p_{\min})(N\vee J)> 2\phi^{-1} R\Big)\\
        = & 1,
 \end{split}
  \end{equation}
  % {\color{red} do we need mention $\phi^{-1}=O(1)$ in the theorem?}
  where the last line is obtained according to Lemma~\ref{lemma:R-bound} and that $\xi_{N,J}\to\infty$ in probability. Similarly,
  \begin{equation}
  \begin{split}
      \lim_{N,J\to\infty} \Pr\Big(\inf_{K^*+2\leq K\leq K_{\max}}u(n,N,J,K)>2\phi^{-1} R\Big)
        \geq & \liminf_{N,J\to\infty} \Pr\Big(\xi_{N,J}\Big(\frac{p_{\max}}{p_{\min}}\Big)(N\vee J)>2 \phi^{-1}R\Big)\\
        =&1.
  \end{split}
  \end{equation}
  Thus, conditions of Lemma~\ref{lemma:overselect} are verified and we obtain
  \begin{equation}\label{eq:ic-general-final-1}
    \lim_{N,J\to\infty}\Pr(\hat{K}> K^*)=0.
  \end{equation}
  Next, we verify conditions of Lemma~\ref{lemma:underselect}.
  According to Lemma~\ref{lemma:R-bound} and the assumption $p_{\min}^{-2}p_{\max} K^*(N\vee J)=o(\sigma_{K^*+1}^2(M^*))$, we have
\begin{equation}
  4(\delta_{C^2}p_{\min})^{-1}K^*R  = O_p\big(p_{\min}^{-2}p_{\max} K^*(N\vee J)\big)=o_p\big(\sigma_{K^*}^2(M^*)\big).
\end{equation}
Thus,
\begin{equation}\label{eq:verifying-under}
  \lim_{N,J\to\infty}\Pr\Big(4(\delta_{C^2}p_{\min})^{-1}K^*R \leq \sigma^2_{K^*}(M^*)\Big)=1.
\end{equation}
In addition, according to the first line of \eqref{eq:cases},
\begin{equation}\label{eq:verifying-under2}
  \begin{split}
    &\lim_{N,J\to\infty}  \Pr\Big(
   u(n,N,J,K)< \phi^{-1}\delta_{C^2}p_{\min}\sigma^2_{K+1}(M^*) \text{ for } 1\leq K\leq K^*
  \Big)\\
  \geq & \liminf_{N,J\to\infty} \Pr\Big(\xi_{N,J}^{-1}p_{\min}\sigma_{K^*+1}^2(M^*)  <  \phi^{-1}\delta_{C^2}p_{\min}\sigma^2_{K+1}(M^*) \text{ for all }K  \\
    \geq & \liminf_{N,J\to\infty} \Pr\Big(\xi_{N,J}^{-1}  < \phi^{-1}\delta_{C^2} \Big)\\
    = & 1.
  \end{split}
\end{equation}
From \eqref{eq:verifying-under} and \eqref{eq:verifying-under2}, conditions of Lemma~\ref{lemma:underselect} are verified and thus
\begin{equation}\label{eq:ic-general-final-2}
    \lim_{N,J\to\infty}\Pr(\hat{K}< K^*)=0.
  \end{equation}
We complete the proof by combining \eqref{eq:ic-general-final-1} and \eqref{eq:ic-general-final-2}.
\end{proof}

\begin{proof}[of Theorem~\ref{thm:IC-simplified}]
First note that the existence of $u$ satisfying \eqref{eq:simplified-ic-condition-u}  implies $N\vee J=o(\sigma_{K^*+1}^2(M^*))$, which further implies $p_{\min}^{-2}p_{\max}K^*(N\vee J)= o(\sigma_{K^*+1}^2(M^*))$ under the asymptotic regime $p_{\min}^{-1}=O(1)$, $K^*=O(1)$. Thus, the assumption about the singular value of $M^*$ in Theorem~\ref{thm:IC-general} is verified. Also, $p_{\min}^{-1}=O(1)$ implies that $(N\wedge J)\log(N+J)= o(n^*/(\log n^*)^2)$. Thus, \eqref{eq:IC-general-asym-reg} is verified.

We proceed to verify that $u$ satisfies \eqref{eq:cases} in Theorem~\ref{thm:IC-general}. We note that
$p_{\min}^{-1}=O(1)$, $K^*=O(1)$ and $u$ satisfies \eqref{eq:simplified-ic-condition-u} implies that there exists $\xi_{N,J}\to\infty$ satisfying
\begin{equation}\label{eq:simp-final1}
  u(n,N,J,K)\leq \xi_{N,J}^{-1}p_{\min}\sigma_{K^*+1}^2(M^*) \text{ for all } K,
\end{equation}
\begin{equation}\label{eq:simp-final2}
   u(n,N,J,K)\geq \xi_{N,J} (p_{\max}/p_{\min}) (N\vee J) \text{ for all }K,
\end{equation}
 and
\begin{equation}\label{eq:simp-final3}
   u(n,N,J,K^*+1)\geq \xi_{N,J} (K^*+1)(p_{\max}/p_{\min}) (N\vee J).
\end{equation}
Note that $\sigma_{K+1}^2(M^*)\geq \sigma_{K^*+1}^2(M^*)$ for $K\leq K^*$. Thus, \eqref{eq:simp-final1} implies the first line of \eqref{eq:cases}; \eqref{eq:simp-final3} implies the second line of \eqref{eq:cases}; \eqref{eq:simp-final2} implies the last line of \eqref{eq:cases}. It verifies \eqref{eq:cases} and completes the proof.
\end{proof}

\begin{proof}[of Corollary~\ref{coro:suggested}]
  Under the asymptotic regime \eqref{eq:IC-simple-asym-reg} and $N\vee J=o(\sigma^2_{K^*+1}(M^*))$, \eqref{eq:IC-general-asym-reg} and $p_{\min}^{-2}p_{\max}K^*(N\vee J)= o(\sigma_{K^*+1}^2(M^*))$ are verified in the proof of Theorem~\ref{thm:IC-simplified}. We now verify \eqref{eq:cases}.

  From the conditions on $h(n,N,J)$, there exists a sequence $\xi_{N,J}$ (possibly depending on $h(n,N,J)$) such that $\xi_{N,J}\to\infty$ in probability and
    \begin{equation}\label{eq:xi-con}
    \xi_{N,J}< h(n,N,J)(p_{\min}/p_{\max})(K^*+1)^{-1} \text{ and } \xi_{N,J}\leq (h(n,N,J))^{-1}(N\vee J)^{-1} p_{\min}\sigma^2_{K^*+1}(M^*).
    \end{equation}
 Also, note that $u(n,N,J,K)=v(n,N,J,K)-v(n,N,J,K-1)=(N\vee J)h(n,N,J)$. It is not hard to verify \eqref{eq:xi-con} implies \eqref{eq:cases}, and, thus Theorem~\ref{thm:IC-general} applies.

 We proceed to the proof of the `in particular' part.  Note that by definition $E(n)=n^*$ and $Var(n)=\sum_{i}\sum_j p_{ij}(1-p_{ij})\leq \sum_i\sum_j p_{ij}=n^*$, which implies $\lim_{N,J\to\infty}\Pr(n>2n^* \text{ or } n<n^*/2)=0$ and further implies $$\lim_{N,J\to\infty}\Pr\big(n/(N\vee J)\geq 2n^*/(N\vee J) \text{ or } n/(N\vee J)\leq n^*/\{2(N\vee J)\}\big)=0.$$
Note that in this part, $h(n,N,J)=\log(n/(N\vee J))$. Also, $\log(n^*/\{2(N\vee J)\})\to\infty$. Thus, $h(n,N,J)\to\infty $ in probability.
In addition, on the event $n/(N\vee J)\leq 2n^*/(N\vee J)$, $(h(n,N,J))^{-1}(N\vee J)^{-1}\sigma_{K^*+1}^2(M^*)\geq \log\big(2n^*/(N\vee J)\big)(N\vee J)^{-1}\sigma_{K^*+1}^2(M^*)$. The right-hand-side of this inequality tend to infinity under the assumptions of the Corollary. This implies
$(h(n,N,J))^{-1}(N\vee J)^{-1}\sigma_{K^*+1}^2(M^*)\to\infty$ in probability.
\end{proof}

\section{Proof of Supporting Lemmas}\label{sec:proof-lemmas}
%{\color{red} Can we start with Section B? Otherwise, the equation numbers are the same as those in the appendix of the main test. }
%{\color{blue} XL: How about the current way of numbering? S. represents supplement, followed by section number.}
%{\color{red} begin proof looks strange. Need check}
\begin{proof}[of Lemma~\ref{lemma:correct-specify-upper}]
By definition,
  \begin{equation}\label{eq:l-start}
\begin{split}
  &\phi\{l(\MM,\YY,\Omega)-l(\MM^*,\YY,\Omega)\}\\
  = & \sum_{ij} \omega_{ij}\big\{y_{ij}m_{ij}-b(m_{ij})- y_{ij}m^*_{ij}+b(m^*_{ij})\big\}\\
  = & \sum_{i,j} (y_{ij}-b'(m^*_{ij}))(m_{ij}-m^*_{ij})\omega_{ij}-\sum_{ij}\big\{b(m_{ij})-b(m_{ij}^*)-b'(m_{ij}^*)(m_{ij}-m^*_{ij})\big\} \omega_{ij}.\\
  %(\omega_{ij}-p)\\
% &~~~ -p \sum_{ij} \big\{b(m_{ij})-b(m_{ij}^*)-b'(m_{ij}^*)(m_{ij}-m^*_{ij})\big\}
\end{split}
\end{equation}
In the rest of the proof, we derive upper  bounds for each term on the right-hand-side of the above display.
% Note that for an exponential family, the function $b$ is always convex. So the first term $b(m_{ij})-b(m_{ij}^*)-b'(m_{ij}^*)(m_{ij}-m^*_{ij})\geq 0$. This implies
For the first term $\sum_{i,j} (y_{ij}-b'(m^*_{ij}))(m_{ij}-m^*_{ij})\omega_{ij}$, we write it as
\begin{equation}
  \sum_{i,j} (y_{ij}-b'(m^*_{ij}))(m_{ij}-m^*_{ij})\omega_{ij}
  = \langle {\ZZ\circ \Omega,\MM-\MM^*} \rangle,
\end{equation}
% \begin{equation}
%   \phi\{l(\MM,\YY,\Omega)-l(\MM^*,\YY,\Omega)\} \leq \sum_{i,j} (y_{ij}-b'(m^*_{ij}))(m_{ij}-m^*_{ij})\omega_{ij}=\langle{\ZZ\circ \Omega,\MM-\MM^*}\rangle,
% \end{equation}
where $\langle{A,B}\rangle = tr(A^TB)$ denotes the matrix inner product. Recall the following inequality in linear algebra: $|\langle {A,B} \rangle |\leq \|A\|_2\|B\|_*\leq \sqrt{\textrm{rank}(B)}\|A\|_{2}\|B\|_F$ for any two matrices $A$ and $B$. Applying this fact to the above display, we obtain
\begin{equation}
\begin{split}
    \big|\sum_{i,j} (y_{ij}-b'(m^*_{ij}))(m_{ij}-m^*_{ij})\omega_{ij}\big|\leq & \{\textrm{rank}(\MM-\MM^*)\}^{1/2}\|\ZZ\circ\Omega\|_2 \|\MM-\MM^*\|_F.\\
  % \leq & \sqrt{\textrm{rank}(\MM)+\textrm{rank}(\MM^*)}\|\ZZ\circ\Omega\|_2 \|\MM-\MM^*\|_F
\end{split}
\end{equation}
Notice that $\textrm{rank}(\MM-\MM^*)\leq r+r^*$ for $\MM\in\mM_r$. Thus, the above inequality implies
\begin{equation}\label{eq:l-bound-1}
\begin{split}
    \big|\sum_{i,j} (y_{ij}-b'(m^*_{ij}))(m_{ij}-m^*_{ij})\omega_{ij}\big|\leq & (r+r^*)^{1/2}\|\ZZ\circ\Omega\|_2 \|\MM-\MM^*\|_F.
  % \leq & \sqrt{\textrm{rank}(\MM)+\textrm{rank}(\MM^*)}\|\ZZ\circ\Omega\|_2 \|\MM-\MM^*\|_F
\end{split}
\end{equation}
We proceed to the analysis of the second term
$\sum_{ij}\big\{b(m_{ij})-b(m_{ij}^*)-b'(m_{ij}^*)(m_{ij}-m^*_{ij})\big\} \omega_{ij}$.
 Note that for $\MM\in \mM_r$, $|m_{ij}|\leq \|B_i\|\|G_j\|\leq C^2$. Similarly, $|m^*_{ij}|\leq C^2$. Thus, for any $\tilde{m}_{ij}=t m_{ij}^*+(1-t)m_{ij}$ and $t\in (0,1)$, $|\tilde{m}_{ij}|\leq C^2$. Recall the definition of $\delta_{C^2}=\inf_{|x|\leq C^2}b''(x)$. Then, $\frac{1}{2}b''(\tilde{m}_{ij})\geq \delta_{C^2}$. This implies
 \begin{equation}\label{eq:second-term-med1}
 \begin{split}
      &\sum_{ij}\big\{b(m_{ij})-b(m_{ij}^*)-b'(m_{ij}^*)(m_{ij}-m^*_{ij})\big\} \omega_{ij}\\
      = & \sum_{ij} \frac{1}{2}b''(\tilde{m}_{ij})(m_{ij}-m^*_{ij})^2\omega_{ij}\\
      \geq & \delta_{C^2} \sum_{ij} (m_{ij}-m^*_{ij})^2\omega_{ij}.
 \end{split}
 \end{equation}
 Note that
 \begin{equation}\label{eq:second-term-med2}
\begin{split}
  &\sum_{ij} (m_{ij}-m^*_{ij})^2\omega_{ij}\\
  = & \sum_{ij} (m_{ij}-m^*_{ij})^2(\omega_{ij}-p_{ij})+\sum_{ij}p_{ij}(m_{ij}-m^*_{ij})^2\\
  \geq &\langle {(\MM-\MM^*)\circ (\MM-\MM^*), Q}\rangle +p_{\min}\|\MM-\MM^*\|_F^2\\
  \geq &- \|(\MM-\MM^*)\circ (\MM-\MM^*)\|_*\|Q\|_2+p_{\min}\|\MM-\MM^*\|_F^2.
\end{split}
\end{equation}
where we define $Q=\Omega-P$ and $P=(p_{ij})_{1\leq i\leq N,1\leq j\leq J}$.
The next lemma is helpful for bounding matrix norms involving Hadamard products, whose proof is given later this section.
\begin{lemma}\label{lemma:Hadamard}
  For $M\in \mM_r$, $\|(\MM-\MM^*)\circ (\MM-\MM^*)\|_*\leq 2C^2(r+r^*)^{1/2}\|\MM-\MM^*\|_F$.
\end{lemma}
\begin{remark}
The proof of Lemma~\ref{lemma:Hadamard} utilizes the property that $m_{ij}=B_i^TG_j$ with $\|B_i\|,\|G_j\|\leq C$ and combine it with a result in \cite{horn1995norm}. This improves the estimate in \cite{chen2019structured} where $|m_{ij}|\leq C^2$ is directly used to derive an upper bound $2C^2 (r+r^*) \|\MM-\MM^*\|_F$. Comparing with this bound, the above lemma provide a sharper bound in the order of $r+r^*$.
\end{remark}
Applying Lemma~\ref{lemma:Hadamard} to \eqref{eq:second-term-med2} and combine it with \eqref{eq:second-term-med1}, we obtain
\begin{equation}\label{eq:l-bound-2.1}
\begin{split}
     & \sum_{ij}\big\{b(m_{ij})-b(m_{ij}^*)-b'(m_{ij}^*)(m_{ij}-m^*_{ij})\big\} \omega_{ij} \\
    \geq &  \delta_{C^2}\big\{p_{\min}\|\MM-\MM^*\|_F^2 - 2C^2(r+r^*)^{1/2}\|\MM-\MM^*\|_F\|Q\|_2\big\}.
\end{split}
\end{equation}
We complete the proof by combining the above display with \eqref{eq:l-start} and \eqref{eq:l-bound-1}.
\end{proof}
\begin{proof}[of Lemma~\ref{lemma:Hadamard}]
  Let $\tilde{B}_i=(B_i^T,-(B_i^*)^T)^T$ and $\tilde{G}_j=(G_j^T,-(G_j^*)^T)^T$. Then, $\tilde{B}_i,\tilde{G}_j\in \mathbb{R}^{r+r^*}$, $\|\tilde{B}_i\|,\|\tilde{G}_j\|\leq \sqrt{2}C$, and $m_{ij}-m^*_{ij}=\tilde{B}_i^T\tilde{G}_j$  for all $i,j$.

  On the other hand, Theorem 2 in \cite{horn1995norm} states that, for any $m\times n$ matrices $A=(a_{ij}),B=(b_{ij})$, if $a_{ij}=g_j^Tf_i$ for vectors $g_j$ and $f_i$s. Then,
  \begin{equation}
    \sum_{i=1}^k \sigma_{i}(A\circ B)\leq \sum_{i=1}^k \|f_{[i]}\|\|g_{[i]}\|\sigma_i(B)\text{ for } k=1,\cdots, m\wedge n,
  \end{equation}
  where $\sigma_i(\cdot)$ denotes the $i$th largest singular value of a matrix,
  $\|f_{[1]}\|\geq \|f_{[2]}\|\geq\cdots\geq \|f_{[m]}\|$ and $\|g_{[1]}\|\geq\cdots\geq\|g_{[n]}\|$ denote the order statistics of $\{\|f_i\|\}_{i=1}^m$ and $\{\|g\|_j\}_{j=1}^n$.
Now, we let $k=N\wedge J$, $A=M-M^*$, $B=A$, $f_{i}=\tilde{B}_i$, $g_{j}=\tilde{G}_j$ in the above result and note that $\|f_{[i]}\|,\|g_{[j]}\|\leq \sqrt{2}C$ in this case, we obtain
\begin{equation}
  \sum_{i=1}^{N\wedge J} \sigma_{i}((M-M^*)\circ (M-M^*))\leq \sum_{i=1}^{N\wedge J} 2C^2 \sigma_i (M-M^*)=2C^2\|M-M^*\|_* .
\end{equation}
Noting the left-hand side of the above display equals $\|(M-M^*)\circ (M-M^*)\|_*$. Thus,
\begin{equation}
  \|(M-M^*)\circ (M-M^*)\|_*\leq 2C^2\|M-M^*\|_*\leq 2C^2(r+r^*)^{1/2} \|M-M^*\|_F.
\end{equation}

\end{proof}

The proofs of Lemmas~\ref{lemma:spectral-pattern} and \ref{lemma:spectral-subexpnential-with-missing} are based on the next lemma that  provides an upper tail bound for the spectral norm of a large class of random matrices. Its proof mainly combines standard symmetrization and truncation arguments with a recent result by \cite{bandeira2016sharp} on the spectral norm of symmetric random matrices with independent, centered and symmetric entries.
\begin{lemma}\label{lemma:general-spectral-bound}
  Let $X=(x_{ij})_{1\leq i\leq N,1\leq j\leq J}$ be an $N\times J$ matrix with $E(x_{ij})=0$ and $E(x^2_{ij})<\infty$.  Then, there is a universal constant $c>0$ such that for all $t,\lambda\geq 0$
  \begin{equation}
  \begin{split}
        \Pr\Big(\|X\|_2\geq 4(\sigma_1\vee\sigma_2)+t\Big)
        \leq  (N+J)e^{-t^2/(c\lambda^2)}+\sum_{i=1}^N\sum_{j=1}^J\Pr\big(|x_{ij}-x_{ij}'|> \lambda\big),
  \end{split}
  \end{equation}
  where we define $\sigma_1=\max_{1\leq i\leq N}\{\sum_{j=1}^J E(x^2_{ij})\}^{1/2}$, $\sigma_2=\max_{1\leq j\leq J}\{\sum_{i=1}^N E(x^2_{ij})\}^{1/2}$, and $x_{ij}'$ is an independent copy of $x_{ij}$.
\end{lemma}

\begin{proof}[of Lemma~\ref{lemma:general-spectral-bound}]
  Let $X'= (x'_{ij})$ which is an independent copy of $X$ and let $\tilde{X}=(\tilde{x}_{ij})=X-X'$. Then, $\tilde{x}_{ij}$s have symmetric distribution and are independent. Let $Z=(z_{ij})=\begin{pmatrix}
    0 & \tilde{X}\\
    \tilde{X}^T & 0
  \end{pmatrix}.$ $Z$ is a symmetric $(N+J)\times (N+J)$ random matrix whose entries are independent and symmetric random variables. Define a random matrix $Z(\lambda)$ as the truncated $Z$,
  \begin{equation}
    Z(\lambda) = (z_{ij}(\lambda))_{1\leq i\leq N,1\leq j\leq J}= (z_{ij}I(|z_{ij}|\leq \lambda))_{1\leq i\leq N, 1\leq j\leq J}.
  \end{equation}
  Then, entries of $Z({\lambda})$ are independent, symmetric random variables and are bounded by $\lambda$. Apply Corollary 3.12 in \cite{bandeira2016sharp}  to $Z(\lambda)$, then there exists a universal constant $c>0$ such that
  \begin{equation}
    \Pr\Big(\|Z(\lambda)\|_2\geq 2^{3/2} \max_{1\leq i\leq (N+J)}\big[\sum_{j=1}^{N+J} E\{z^2_{ij}(\lambda)\}\big]^{1/2}+t\Big)\leq (N+J) e^{-t^2/(c\lambda^2)}
  \end{equation}
  Note that
  \begin{equation}
  \begin{split}
      \max_{1\leq i\leq (N+J)}\big[\sum_{j=1}^{N+J} E\{z^2_{ij}(\lambda)\}\big]^{1/2}\leq & \max_{1\leq i\leq (N+J)}\big\{\sum_{j=1}^{N+J} E(z^2_{ij})\big\}^{1/2}\\
      = & \max\big[\max_{1\leq i\leq N}\{\sum_{j=1}^J E(\tilde{x}^2_{ij})\}^{1/2},\max_{1\leq j\leq J}\{\sum_{i=1}^N E(\tilde{x}^2_{ij})\}^{1/2}\big]\\
      = & \sqrt{2}(\sigma_1\vee\sigma_2).
  \end{split}
  \end{equation}
  Thus,
  \begin{equation}
    \Pr\Big(\|Z(\lambda)\|_2\geq 4(\sigma_1\vee \sigma_2)+t\Big)\leq (N+J) e^{-t^2/(c\lambda^2)}.
  \end{equation}
  On the other hand,
  \begin{equation}
  \begin{split}
    \Pr\big(\|Z\|_2\geq 4(\sigma_1\vee \sigma_2)+t\big)
    %\\
    % \leq & \Pr(\|Z(\lambda)\|_2\geq 2(\sigma_1\vee \sigma_2)+t)+\Pr(Z\neq Z(\lambda))\\
    \leq  \Pr\big(\|Z(\lambda)\|_2\geq 4(\sigma_1\vee \sigma_2)+t\big)+\Pr\big(\max_{1\leq i,j\leq N+J}|z_{ij}|>\lambda\big).
  \end{split}
  \end{equation}
  The above two inequalities together imply
  \begin{equation}
    \Pr\big(\|Z\|_2\geq 4(\sigma_1\vee \sigma_2)+t\big)
    \leq  (N+J) e^{-t^2/(c\lambda^2)} + \Pr(\max_{1\leq i,j\leq N+J}|z_{ij}|>\lambda).
  \end{equation}
  Note that $\|Z\|_2=\|\tilde{X}\|_2$ and $\max_{1\leq i,j\leq N+J}|z_{ij}| = \max_{1\leq i\leq N,1\leq j\leq J}|\tilde{x}_{ij}|$. From the above inequality, we obtain
  \begin{equation}
    \Pr\big(\|\tilde{X}\|_2\geq 4(\sigma_1\vee \sigma_2)+t\big)
    \leq  (N+J) e^{-t^2/(c\lambda^2)} + \Pr(\max_{1\leq i\leq N, 1\leq j\leq J}|\tilde{x}_{ij}|>\lambda).
  \end{equation}
  With a union bound, we further get
  \begin{equation}\label{eq:bound-tilde-x}
    \Pr(\|\tilde{X}\|_2\geq 4(\sigma_1\vee \sigma_2)+t)
    \leq  (N+J) e^{-t^2/(c\lambda^2)} + \sum_{1\leq i\leq N,1\leq j\leq J}\Pr(|\tilde{x}_{ij}|>\lambda).
  \end{equation}
  Recall $\tilde{X}= X-X'$ and the function $I(\|X-X'\|_2\geq 4(\sigma_1\vee\sigma_2)+t)$ is convex in $X'$. Thus, by Jensen's inequality,
  \begin{equation}
    \Pr(\|X\|_2\geq 4(\sigma_1\vee\sigma_2)+t) \leq \Pr(\|X-X'\|_2\geq 4(\sigma_1\vee\sigma_2)+t) = \Pr(\|\tilde{X}\|_2\geq 4(\sigma_1\vee\sigma_2)+t).
  \end{equation}
  This, together with \eqref{eq:bound-tilde-x} completes the proof.
\end{proof}

\begin{proof}[of Lemma~\ref{lemma:spectral-pattern}]
Let $\omega'_{ij}$ be an independent copy of $\omega_{ij}$, then $|\omega'_{ij}-p_{ij}-(\omega_{ij}-p_{ij})|\leq 1$. In addition, $E(\omega_{ij}-p_{ij})^2 = p_{ij}(1-p_{ij})\leq p_{ij}$. Thus, $\max_{i}\{\sum_{j}E(\omega_{ij}-p_{ij})^2\}^{1/2}\leq \max_{i} (\sum_j p_{ij})^{1/2}=(\max_in^*_{i\cdot})^{1/2} $ and $\max_{j}\{\sum_{i}E(\omega_{ij}-p_{ij})^2\}^{1/2}\leq (\max_j n^*_{\cdot j})^{1/2} $.

Choose $\lambda=1$ and apply Lemma~\ref{lemma:general-spectral-bound} to $\Omega-P$, we obtain that for all $t\geq 0$,
\begin{equation}
  \Pr\big(\|\Omega-P\|_2\geq 4 (\max_in^*_{i\cdot}\vee \max_jn^*_{\cdot j})^{1/2} +t\big)\leq (N+J)e^{-t^2/c}.
\end{equation}
Let $t=(2c \log(N+J))^{1/2}$ in the above inequality, we obtain
\begin{equation}
  \Pr\big(\|\Omega-P\|_2\geq 4 (\max_in^*_{i\cdot}\vee \max_jn^*_{\cdot j})^{1/2}+(2c \log(N+J))^{1/2}\big)\leq (N+J)^{-1}.
\end{equation}
We complete the proof by noting that $(2c)^{1/2}$ is still a universal constant.
\end{proof}

\begin{proof}[of Lemma~\ref{lemma:spectral-subexpnential-with-missing}]
  Apply Lemma~\ref{lemma:general-spectral-bound} to $V\circ\Omega$, we obtain that for all $t,\lambda\geq 0$,
  \begin{equation}\label{eq:z-circ-start}
    \Pr\big(\|V\circ \Omega\|_2\geq 4(\sigma_1\vee \sigma_2)+t\big)
    \leq (N+J)e^{-t^2/(c\lambda^2)}+\sum_{ij} \Pr(|v_{ij}\omega_{ij}-v'_{ij}\omega'_{ij}|\geq \lambda),
  \end{equation}
  where $(v'_{ij},\omega'_{ij})$ is an independent copy of $(v_{ij},\omega_{ij})$, $\sigma_1=\max_{i}\{\sum_jE(v^2_{ij}\omega^2_{ij})\}^{1/2}$ and $\sigma_2=\max_{j}\{\sum_iE(v^2_{ij}\omega^2_{ij})\}^{1/2}$. We proceed to a detailed analysis of $\sigma_1,\sigma_2$ and the probability $\Pr(|v_{ij}\omega_{ij}-v'_{ij}\omega'_{ij}|\geq \lambda)$. First, a direct calculation gives
  \begin{equation}
    \sigma_1=\max_i\big\{\sum_{j}p_{ij}E(v_{ij}^2)\big\}^{1/2}\leq (\max_{i} n^*_{i\cdot})^{1/2} \max_{ij}\{E(v_{ij}^2)\}^{1/2}.
  \end{equation}
  Similarly, $\sigma_2\leq (\max_{j} n^*_{\cdot j})^{1/2} \max_{ij}\{E(v_{ij}^2)\}^{1/2}$.
  Now we find an upper bound of $\Pr(|v_{ij}\omega_{ij}-v'_{ij}\omega'_{ij}|\geq \lambda)$.
  Note that
  \begin{equation}\label{eq:entry-bound}
  \begin{split}
     &\Pr(|v_{ij}\omega_{ij}-v'_{ij}\omega'_{ij}|\geq \lambda)\\
    = & p_{ij}^2 \Pr(|v_{ij}-v'_{ij}|\geq \lambda) + 2p_{ij}(1-p_{ij}) \Pr(|v_{ij}|\geq \lambda)\\
     \leq & 3p_{ij} \Pr(|v_{ij}-v'_{ij}|\geq \lambda) \vee \Pr(|v_{ij}|\geq \lambda).
  \end{split}
  \end{equation}
  For $\Pr(|v_{ij}|\geq \lambda)$, we use a tail bound for sub-exponential variables
  \begin{equation}
    \Pr(|v_{ij}|\geq \lambda)\leq 2 e^{-\lambda^2/(2\nu^2)}\vee e^{-\lambda/(2\alpha)}.
  \end{equation}
  Similarly, noting that $v_{ij}-v'_{ij}$ is also sub-exponential with parameters $2\nu^2,\alpha$, we have
  \begin{equation}
    \Pr(|v_{ij}-v'_{ij}|\geq \lambda)
    \leq 2 e^{-\lambda^2/(4\nu^2)}\vee e^{-\lambda/(2\alpha)}.
  \end{equation}
  Combining the above two inequalities with \eqref{eq:entry-bound}, we have
  \begin{equation}
    \Pr(|v_{ij}\omega_{ij}-v'_{ij}\omega'_{ij}|\geq \lambda)
    \leq 6p_{ij} e^{-\lambda^2/(4\nu^2)}\vee e^{-\lambda/(2\alpha)}.
  \end{equation}
  Combining the above inequality with \eqref{eq:z-circ-start}, we arrive at
  \begin{equation}\label{eq:choosing-l-and-t}
  \begin{split}
        &\Pr\big(\|V\circ \Omega\|_2\geq 4\max_{ij}\{E(v_{ij}^2)\}^{1/2}(\max_{i} n^*_{i\cdot})^{1/2}\vee (\max_{j} n^*_{\cdot j})^{1/2}  +t\big)\\
        \leq &(N+J)e^{-t^2/(c\lambda^2)}+ 6e^{-\lambda^2/(4\nu^2)}\vee e^{-\lambda/(2\alpha)}n^*.
  \end{split}
  \end{equation}
Let $\lambda= 4(\alpha\vee\nu)\log n^*$. It is not hard to verify that
$  6e^{-\lambda^2/(4\nu^2)}\vee e^{-\lambda/(2\alpha)}n^*\leq (n^*)^{-1}$
 for $n^*\geq 6$. Let $t=\lambda \{2c \log(N+J)\}^{1/2}$, we obtain
$
  (N+J)e^{-t^2/(c\lambda^2)}\leq (N+J)^{-1}
$.
Combining the above inequalities with \eqref{eq:choosing-l-and-t}, we obtain that with probability at least $1-(N+J)^{-1}-(n^*)^{-1}$,
\begin{equation}
  \begin{split}
        &\|V\circ \Omega\|_2\\
        \leq & 4\max_{ij}\{E(v_{ij}^2)\}^{1/2}(\max_{i} n^*_{i\cdot})^{1/2}\vee (\max_{j} n^*_{\cdot j})^{1/2}   + 4\sqrt{2} c^{1/2} (\alpha\vee \nu)\log n^* \log^{1/2}(N+J).
  \end{split}
  \end{equation}
  This completes the proof of inequality \eqref{eq:spectral-subexp} (note that $4\sqrt{2} c^{1/2}$ is also a universal constant). We proceed to prove the `in particular' part of the lemma.
  For each $z_{ij}=y_{ij}-b'(m^*_{ij})$, its second moment is $E(z_{ij}^2)=\phi b''(m^*_{ij})\leq \phi\kappa_{2C^2}$. In addition, its moment generating function is $E(e^{\lambda z_{ij}}) = \exp[\phi^{-1}\{b(m_{ij}^*+\lambda \phi)-b(m^*_{ij})\}-\lambda  b'(m^*_{ij})] = \exp\{\phi b''(m^*_{ij}+\tilde{\lambda}\phi)\lambda^2 /2 \}$ for some $|\tilde{\lambda}|\leq |\lambda|$. Since $|m_{ij}^*|\leq C^2$ by assumption, we can see that for $|\lambda|\leq C^2/\phi$, $|m^*_{ij}+\tilde{\lambda}\phi|\leq 2C^2$ and thus
$E(e^{\lambda z_{ij}}) \leq \exp\{\kappa_{2C^2}\phi\lambda^2 /2 \}$ for all $|\lambda| \leq C^2/\phi$. This implies that $z_{ij}$ is sub-exponential with the parameters $\nu^2=\phi \kappa_{2C^2}$ and $\alpha=\phi/C^2$. We complete the proof by applying \eqref{eq:spectral-subexp} with the above parameters for $Z$.
\end{proof}
\bigskip

\begin{proof}[of Lemma~\ref{lemma:overselect}]
For each $K^*+1\leq K\leq K_{\max}$, we first derive an upper bound for $\phi\{l(\hat{M}^{(K)},Y,\Omega)-l(\hat{M}^{(K^*)},Y,\Omega)\}-(v(n,N,J,K)-v(n,N,J,K^*))$. According to Lemma~\ref{lemma:correct-specify-upper},
\begin{equation}
\begin{split}
     &\phi\{l(\hat{M}^{(K)},Y,\Omega)-l(M^*,Y,\Omega)\}\\
     \leq &
(K+K^*+2)^{1/2}\big(\|\ZZ\circ\Omega\|_2+ 2\delta_{C^2}C^2\|Q\|_2\big)\|\hat{\MM}^{(K)}-\MM^*\|_F\\
\leq &  2K^{1/2}\big(\|\ZZ\circ\Omega\|_2+ 2\delta_{C^2}C^2\|Q\|_2\big)\|\hat{\MM}^{(K)}-\MM^*\|_F.
\end{split}
 \end{equation}
Combining this with \eqref{eq:error-bound-mle0} gives
\begin{equation}
    \phi\{l(\hat{M}^{(K)},Y,\Omega)-l(M^*,Y,\Omega)\}\leq 4p_{\min}^{-1}K \big\{\|\ZZ\circ\Omega\|_2+ 2\delta_{C^2}C^2\|Q\|_2\big\}^2=KR.
\end{equation}
Thus, the penalized log-likelihood satisfies
\begin{equation}
\begin{split}
   &\max_{K^*+1\leq K\leq K_{\max}}\Big[-2l(\hat{M}^{(K)},Y,\Omega)+v(n,N,J,K)-\{-2l(M^*,Y,\Omega)+v(n,N,J,K^*)\}\Big]\\
   \geq & \max_{K^*+1\leq K\leq K_{\max}}\Big[- 2\phi^{-1}KR + \sum_{l=K^*+1}^K u(n,N,J,K)\Big]
\end{split}
\end{equation}
It is easy to see that, if the events
$u(n,N,J,K^*+1)>2\phi^{-1}(K^*+1)R$ and $u(n,N,J,l)>2\phi^{-1}R$ happen at the same time for all $K^*+2\leq l\leq K_{\max}$, then the right-hand side of the above inequality is strictly greater than zero. Thus,
\begin{equation}
\begin{split}
&\Pr\Big(\hat{K}\leq K^*\Big)\\
 \geq &\Pr\Big(\max_{K^*+1\leq K\leq K_{\max}}\Big[-2l(\hat{M}^{(K)},Y,\Omega)+v(n,N,J,K)-\{-2l(M^*,Y,\Omega)+v(n,N,J,K^*)\}\Big]>0\Big)\\
   \geq & \Pr\Big(u(n,N,J,K^*+1)>2\phi^{-1}(K^*+1)R, \text{ and } \inf_{K^*+2\leq l\leq K_{\max}}u(n,N,J,K)>2\phi^{-1}R \Big).
\end{split}
\end{equation}
We complete the proof by noting the right-hand side of the above inequality tend to one under the assumptions of the lemma.
\end{proof}

The proof of Lemma~\ref{lemma:underselect} requires the next lemma.
\begin{lemma}\label{lemma:likelihood-bound-underselect}
%If $4(\delta_{C^2}p_{\min})^{-1}(K^*)^{1/2}\big\{\|\ZZ\circ\Omega\|_2+ 2\delta_{C^2}C^2\|Q\|_2\big\}\leq  \sigma_{K^*}(M^*)$, then
If $4(\delta_{C^2}p_{\min})^{-1}K^* R\leq \sigma^2_{K^*+1}(M^*)$, then
  \begin{equation}
    \phi\{l(\hat{\MM}^{(K)},\YY,\Omega)-l(\hat{\MM}^{(K^*)},\YY,\Omega)\}\leq -\frac{1}{2}\delta_{C^2}p_{\min}\Big\{\sum_{l=K+2}^{K^*+1}\sigma_l^2(M^*)\Big\}
  \end{equation}
  for $0\leq K\leq K^*-1$.
\end{lemma}
\begin{proof}[of Lemma~\ref{lemma:likelihood-bound-underselect}]
First, according to Lemma~\ref{lemma:correct-specify-upper}, $\hat{M}^{(K)}\in\mM_{K+1}$ and $K+1\leq K^*$, we have
  \begin{equation}
 \begin{split}
     &\phi\{l(\hat{M}^{(K)},Y,\Omega)-l(M^*,Y,\Omega)\}\\\leq &
(K+K^*+2)^{1/2}\big\{\|\ZZ\circ\Omega\|_2+ 2\delta_{C^2}C^2\|Q\|_2\big\}\|\hat{M}^{(K)}-\MM^*\|_F- \delta_{C^2}p_{\min}\|\hat{M}^{(K)}-\MM^*\|_F^2\\
\leq & \sup_{M\in\mM_{K+1}}\Big[2(K^*)^{1/2}\big\{\|\ZZ\circ\Omega\|_2+ 2\delta_{C^2}C^2\|Q\|_2\big\}\|M-\MM^*\|_F- \delta_{C^2}p_{\min}\|M-\MM^*\|_F^2\Big].
 \end{split}
 %   (K+K^*)^{1/2}\phi\|Z\circ\Omega\|_2\|M-M^*\|_F,
 \end{equation}
Note that the expression inside `$\sup$' is a quadratic function in $\|M-M^*\|_F$.
Let $d(M^*,\mM_{K+1}) = \inf_{M\in \mM_{K+1}}\|M-M^*\|_F$. From properties of a quadratic function, if {$d(M^*,\mM_{K+1})\geq 2(\delta_{C^2}p_{\min})^{-1}\cdot2(K^*)^{1/2}\big\{\|\ZZ\circ\Omega\|_2+ 2\delta_{C^2}C^2\|Q\|_2\big\}$}
\begin{equation}
  \begin{split}
         &\phi\{l(\hat{M}^{(K)},Y,\Omega)-l(M^*,Y,\Omega)\}\\\leq &
(2K^*)^{1/2}\big\{\|\ZZ\circ\Omega\|_2+ 2\delta_{C^2}C^2\|Q\|_2\big\}d(M^*,_{K+1})- \delta_{C^2}p_{\min}d^2(M^*,_{K+1})\\
\leq & -\frac{1}{2}\delta_{C^2}p_{\min}d^2(M^*,\mM_{K+1}).
  \end{split}
\end{equation}
Note that $\phi\{l(\hat{M}^{(K^*)},Y,\Omega)-l(M^*,Y,\Omega)\}\geq 0$. Thus, the above inequality implies that on the event {$d(M^*,\mM_K)\geq 4(\delta_{C^2}p_{\min})^{-1}(K^*)^{1/2}\big\{\|\ZZ\circ\Omega\|_2+ 2\delta_{C^2}C^2\|Q\|_2\big\}$},
\begin{equation}\label{eq:underselct-mid-bound}
  \phi\{l(\hat{M}^{(K)},Y,\Omega)-l(\hat{M}^{(K^*)},Y,\Omega)\}\leq -\frac{1}{2}\delta_{C^2}p_{\min}d^2(M^*,\mM_K).
\end{equation}
Now we proceed to a lower bound for $d(M^*,\mM_{K+1})$. Recall the well-known fact that
$\inf\limits_{M \text{ has a rank } K+1}\|M^*-M\|^2_F=\sum_{l=K+2}^{K^*+1}\sigma_l^2(M^*)$ where $\sigma_1(M^*)\geq \cdots\geq \sigma_{K^*+1}(M^*)$ denotes the non-zero singular values of $M^*$. Thus, $d(M^*,\mM_{K+1})\geq \{\sum_{l=K+2}^{K^*+1}\sigma_l^2(M^*)\}^{1/2} \geq \sigma_{K^*+1}(M^*)$. Combine this with \eqref{eq:underselct-mid-bound}, we have
\begin{equation}
   \phi\{l(\hat{M}^{(K)},Y,\Omega)-l(\hat{M}^{(K^*)},Y,\Omega)\}\leq -\frac{1}{2}\delta_{C^2}p_{\min}\Big\{\sum_{l=K+2}^{K^*+1}\sigma_l^2(M^*)\Big\},
\end{equation}
if $4(\delta_{C^2}p_{\min})^{-1}(K^*)^{1/2}\big\{\|\ZZ\circ\Omega\|_2+ 2\delta_{C^2}C^2\|Q\|_2\big\}\leq  \sigma_{K^*+1}(M^*)$ and $K\leq K^*-1$.
We complete the proof by noting that $4(\delta_{C^2}p_{\min})^{-1}(K^*)^{1/2}\big\{\|\ZZ\circ\Omega\|_2+ 2\delta_{C^2}C^2\|Q\|_2\big\}\leq  \sigma_{K^*+1}(M^*)$ is equivalent to $4(\delta_{C^2}p_{\min})^{-1}K^* R\leq \sigma^2_{K^*+1}(M^*)$.
\end{proof}

\begin{proof}[of Lemma~\ref{lemma:underselect}]
According to Lemma~\ref{lemma:likelihood-bound-underselect}, for each $0\leq K\leq K^*-1$,
 \begin{equation}
  \begin{split}
        &-2l(\hat{\MM}^{(K)},\YY,\Omega)+v(n,N,J,K)-\big\{-2l(\hat{\MM}^{(K^*)},\YY,\Omega)+v(n,N,J,K^*)\big\}\\
         \geq & \phi^{-1}\delta_{C^2}p_{\min}\Big\{\sum_{l=K+2}^{K^*+1}\sigma_l^2(M^*)\Big\} - \sum_{l=K+1}^{K^*}u(n,N,J,l),
  \end{split}
  \end{equation}
  if $4(\delta_{C^2}p_{\min})^{-1}K^*R \leq \sigma^2_{K^*}(M^*)$. Clearly, right-hand-side of the above inequality is strictly greater than zero if $u(n,N,J,l)< \phi^{-1}\delta_{C^2}p_{\min}\sigma_{l+1}^2(M^*)$ for all $1\leq l\leq K^*$. Thus,
  \begin{equation}
  \begin{split}
  & \Pr\big(\hat{K}\geq K^*\big)\\
      \geq & \Pr\Big(\max_{1\leq K\leq K^*}\big[-2l(\hat{\MM}^{(K)},\YY,\Omega)+v(n,N,J,K)-\big\{-2l(\hat{\MM}^{(K^*)},\YY,\Omega)+v(n,N,J,K^*)\big\}\big] >0 \Big)\\
        \geq   & \Pr\big(4(\delta_{C^2}p_{\min})^{-1}K^*R \leq \sigma^2_{K^*+1}(M^*) \text{ and }
   u(n,N,J,K)< \phi^{-1}\delta_{C^2}p_{\min}\sigma^2_{K+1}(M^*) \text{ for all } 1\leq K\leq K^* \big)
  \end{split}
  \end{equation}
  The right-hand-side of the above inequality tend to one under the assumptions of the Lemma. This completes the proof.
\end{proof}

\begin{proof}[of Lemma~\ref{lemma:R-bound}]
According to Lemma~\ref{lemma:spectral-subexpnential-with-missing}, there is a universal constant $c$ such that with probability least $1-(N+J)^{-1}-(n^*)^{-1}$,
   \begin{equation}\label{eq:spectral-simp-bound1}
     \|Z\circ \Omega\|_2\leq  4(\phi\kappa_{2C^2})^{1/2}(\max_{i} n_{i\cdot}^*)^{1/2}\vee (\max_{j} n_{\cdot j}^*)^{1/2}  + c \{(\phi/C^2)\vee (\phi\kappa_{2C^2})^{1/2}\}\log n^* \log^{1/2}(N+J).
   \end{equation}
Under the asymptotic regime \eqref{eq:IC-general-asym-reg}, we have $4(\phi\kappa_{2C^2})^{1/2}=O(1)$, $\max_{i} n_{i\cdot}^* = O(p_{\max} J)$, $\max_{j} n_{\cdot j}^*=O(p_{\max}N)$, $c \{(\phi/C^2)\vee (\phi\kappa_{2C^2})^{1/2}\}=O(1)$, and $\log n^* \log^{1/2}(N+J)=O\big( (N\wedge J)^{-1/2}(n^*)^{1/2}\big)=O\big( \{p_{\max} (N\vee J)\}^{1/2} \big)$. Thus, the right-hand-side of \eqref{eq:spectral-simp-bound1} is of the order $O(\{p_{\max} (N\vee J)\}^{1/2})$ and
\begin{equation}\label{eq:spectral-simp-bound2}
  \|Z\circ \Omega\|_2=O_p(\{p_{\max} (N\vee J)\}^{1/2})
\end{equation}
as $N,J\to\infty$.
Similarly, according to Lemma~\ref{lemma:spectral-pattern},
\begin{equation}
   \|Q\|_2 \leq 4 (\max_{i} n^*_{i\cdot})^{1/2} \vee (\max_{j} n^*_{\cdot j})^{1/2}+c \log^{1/2}(N+J)
  \end{equation}
with probability at least $1-(N+J)^{-1}$. Under the the asymptotic regime \eqref{eq:IC-general-asym-reg}, the right-hand-side of the above inequality is of the order $O(\{p_{\max} (N\vee J)\}^{1/2})$, and thus
\begin{equation}\label{eq:spectral-simp-bound3}
  \|Q\|_2 =O_p(\{p_{\max} (N\vee J)\}^{1/2}).
\end{equation}
We complete the proof by combining \eqref{eq:spectral-simp-bound2}, \eqref{eq:spectral-simp-bound3}, and the definition of $R$.
\end{proof}
% {\color{blue}
% \begin{remark}
%   Combining the above three lemmas, for model selection consistency, we need with high probability
%   \begin{enumerate}
%     \item $u(n,N,J,K^*+1)\gg 2 (K^*+1)(N\vee J)$,
%     \item $u(n,N,J,K^*+l)\gg 2(N\vee J)$ for all $l\geq 2$,
%     \item $u(n,N,J,K)\ll p_{\min} \sigma_{K}^2 $ for $1\leq K\leq K^*$.
%     \item and the model assumption $K^*(N\vee J)\ll  p_{\min} \sigma_{K^*}^2$ plus assumptions for Theorem~\ref{thm:error-bound-simple}.
%   \end{enumerate}

% \end{remark}
% How should we write the theorem? I'm considering giving a simple form first then give a more complicated one, but even the simple one does not look neat.
% \begin{itemize}
%   \item Assume $n^*$ is known? Can simplify `with high probability' part.
%   \item Note $u$ is the difference of $v$. To satisfy 1, we can either choose $v\sim K^2(N\vee J)$ or $v\sim K_{\max} K (N\vee J)$ assuming there is a $K_{\max}$. (can also make it larger). Which one do we want?
% \end{itemize}
% }
{%\color{blue}
\section{Proof of Proposition~\ref{prop:lower-bound}}

Without loss of generality, assume $J\geq N$,  $N/K^*$ is an integer, and $\phi=1$. The proof can be easily extended to the other cases.

To prove the lower bound for the minimax risk, we use a local Fano's method, which is a standard tool for proving lower error bounds \citep{tsybakov2008introduction}.
Throughout the proof we use the notation $F=(F_1,\cdots, F_N)^T$, $d=(d_1,\cdots,d_J)^T$ and $A=(A_1,\cdots, A_J)^T$.

We start with constructing a local packing of as follows. First, let $F^{(0)}=C(I_{K^*},\cdots, I_{K^*})^T$. Note that here we used the assumption that $N/{K^*}$ is an integer. Also, let $d^{(0)}=(0,\cdots, 0)^T$.
Next, according to the (Gilbert-Varshanmov bound) \citep{gilbert1952comparison}, there exists a set $\mathcal{B}=\{B^{(l)}: l=1,\cdots L\}\subset \{1,-1\}^{J\times {K^*}}$ satisfying $L\geq \exp(JK^*/8)$ and
\begin{equation}\label{eq:construct}
  \sum_{j=1}^J\sum_{k=1}^{K^*} I(b_{jk}\neq b'_{jk})\geq JK^*/4
\end{equation}
for any $B,B'\in\mathcal{B}$ and $B\neq B'$.
Then, we construct a set $\mathcal{A}=\{A=\gamma B: B\in \mathcal{B}\}$ for some $\gamma$ specified in the sequel.
Now define
\begin{equation}
\begin{split}
  \mathcal{M}^* &= \{M=(m_{ij}): m_{ij}=d_j+F_i^T A_j \text{ for all } i \text{ and } j \text{ where } d=d^{(0)}, F=F^{(0)} \text{ and } A\in \mathcal{A}\}\\
  & = \{M=F^{(0)}A^T: A\in \mathcal{A}  \}
\end{split}
\end{equation}
The set $\mathcal{M}^*$ defined above has the following properties.
\begin{itemize}
  \item [(a)] $|\mathcal{M}^*|=L\geq \exp(JK^*/8)$.
  \item [(b)]  The Kullback-Leibler divergence $\max_{M,M'\in\mathcal{M}^*}KL(P_{M}\|P_{M'})\leq \kappa n^*\gamma^2$ for some constant $\kappa$, where $P_{M}$ denotes the probability measure for $(Y_{ij},\omega_{ij})_{1\leq i\leq N, 1\leq j\leq J}$ when the true parameter is $M$.
  \item [(c)]
$\|M- M'\|^2_F\geq C^2 NJ\gamma^2$ for $M,M'\in\mathcal{M}^*$ and $M\neq M'$.
\end{itemize}
Property (a) holds obviously. Property (b) holds because of the following inequalities
\begin{equation}
  \begin{split}
    KL(P_{M}\|P_{M'})
    &= \sum_i\sum_j p_{ij} \{b'(m_{ij})(m_{ij}-m'_{ij})-(b(m_{ij})-b(m'_{ij})\}\\
    &\leq p_{\max} \sum_{i}\sum_j\{b'(m_{ij})(m_{ij}-m'_{ij})-(b(m_{ij})-b(m'_{ij})\}\\
    &\leq \kappa n^*/(NJ) \|M-M'\|_F^2\\
    & = \kappa C^2 n^*/(NJ)\cdot (N/{K^*})\sum_{j}\sum_{k}(a_{jk}-a'_{jk})^2,
  \end{split}
\end{equation}
where we used the construction $M=F^{(0)}A^T$ in the last equation. Note that $|a_{jk}|=|\gamma b_{jk}|=\gamma$ for $A\in\mathcal{A}$. Thus, $(a_{jk}-a'_{jk})^2\leq 4\gamma^2$, which leads to property (b) of the set $\mathcal{M}^*$ for a possibly different $\kappa$.
Property (c) holds for the following reasons. By construction, for $M,M'\in\mathcal{M}^*$ and $M\neq M'$
\begin{equation}
  \begin{split}
    \|M-M'\|_F^2
    & = C^2 N/{K^*}\cdot \sum_j\sum_k (a_{jk}-a'_{jk})^2\\
    & = C^2(N/{K^*})\gamma^2 \sum_j\sum_k (b_{jk}-b'_{jk})^2\\
    & = C^2(N/{K^*})\gamma^2 \cdot 4 \sum_j\sum_k I(b_{jk}\neq b'_{jk})\\
    & \geq C^2(N/{K^*})\gamma^2 JK^*,
  \end{split}
\end{equation}
where the last inequality is due to \eqref{eq:construct}. Thus, property (c) holds.

Now, for an arbitrary estimator $\bar{M}$, define a new estimator $\tilde{\bar{M}}=\arg\min_{W\in \mathcal{M}^*}\|W-\bar{M} \|_F$. It is easy to see that for $M\in \mathcal{M}^* $,$
  \|\tilde{\bar{M}} - M \|_F\leq 2 \|\bar{M}-M\|_F.$ By a version of Fano's inequality, we have
\begin{equation}
  \max_{M^*\in \mathcal{M}^*}P_{M^*}\left(\tilde{\bar{M}}\neq M^*\right)
  \geq 1- \frac{\kappa n^*\gamma^2+1}{\log |\mathcal{M}^*|}\geq
  1-\frac{\kappa n^*\gamma^2+1}{JK/8}.
\end{equation}
Choose $\gamma= \kappa^{-1}(JK/n^*)^{1/2}$ for a possibly different $\kappa$, then for $JK\geq 64$, we have
\begin{equation}
  \max_{M^*\in \mathcal{W}_M}P_{M^*}\left(\tilde{\bar{M}}\neq M^*\right)
  \geq \frac{1}{2}.
\end{equation}
Furthermore, we have
\begin{equation}
\begin{split}
    &\max_{M^*\in \mathcal{M}^*}P_{M^*}\left(\|\bar{M}-M^*\|_F
    \geq (1/2)\cdot C(NJ)^{1/2}\gamma
  \right)\\
  \geq& \max_{M^*\in \mathcal{M}^*}P_{M^*}\left(\|\tilde{\bar{M}}-M^*\|_F
  \geq C
(NJ)^{1/2}\gamma
  \right)\\
  \geq & \max_{M^*\in \mathcal{M}^*}P_{M^*}\left(\tilde{\bar{M}}\neq M^*\right)\\
  \geq &\frac{1}{2}.
\end{split}
\end{equation}
Simplifying the term $(1/2)\cdot C(NJ)^{1/2}\gamma$, we arrive at
\begin{equation}
  \max_{M^*\in\mathcal{M}^*}P_{M^*} \left((NJ)^{-1/2}\|\bar{M}-M^*\|_F\geq 2^{-1}\kappa^{-1}\cdot (JK/n^*)^{1/2}
  \right)\geq \frac{1}{2}.
\end{equation}
Note that for $A\in\mathcal{A}$,
$\|A_j\|\leq \gamma\sqrt{{K^*}}= \kappa^{-1}(J(K^*)^2/n^*)^{1/2}$. Thus, for a possibly larger constant $\kappa$, we have $\|A_j\|\leq C$ under the assumption $(K^*)^2(J+N)\leq n^*$. Thus, $\mathcal{M}^*$ is a subset of the parameter space of interest. That is,
\begin{equation}
  \mathcal{M}^*\subset \mathcal{G}:= \{M=(m_{ij}): m_{ij}=d_j+F_i^TA_j, \text{ and }(\Vert F_i \Vert^2 + 1)^{\frac{1}{2}} \leq C, ((d_j)^2 + \Vert A_j \Vert^2)^{\frac{1}{2}}  \leq C, \text{ for all }i\}
\end{equation}
This further implies
\begin{equation}
    \max_{M^*\in \mathcal{G}}P_{M^*} \left(\frac{1}{\sqrt{NJ}}\|\bar{M}-M^*\|_F\geq 2^{-1}\kappa^{-1}\cdot (JK^*/n^*)^{1/2}
  \right)\geq \frac{1}{2}.
\end{equation}
This completes our proof.
}
\section{On Optimization for Joint Likelihood}\label{app:optimization}

We provide some discussions on the optimization problem~\eqref{eq:jml} for the constrained joint maximum likelihood estimator. The two reasons below explain why the solution given by an  alternating maximization algorithm typically performs well, even though \eqref{eq:jml} is a non-convex optimization problem.
First, according to the proofs of Theorems~\ref{thm:error-bound-simple} through \ref{thm:IC-general}, Theorems~\ref{thm:error-bound-simple} and \ref{thm:finite-bound} hold as long as
the estimates satisfy
$$l_K(\hat F_1,..., \hat F_N, \hat A_1, \hat d_1, ..., \hat A_J, \hat d_J) \geq l_{K^*}(F_1^*,..., F_N^*,  A_1^*, d_1^*, ..., A_J^*, d_J^*)$$
when $K\geq K^*$. In addition, for Theorems~\ref{thm:IC-simplified} and \ref{thm:IC-general} to hold, we only need
\begin{equation}\label{eq:non-convex}
l_{K^*}(\hat F_1,..., \hat F_N, \hat A_1, \hat d_1, ..., \hat A_J, \hat d_J) \geq l_{K^*}(F_1^*,..., F_N^*,  A_1^*, d_1^*, ..., A_J^*, d_J^*).
\end{equation}
It means that the number of factors can be consistently selected even if our estimate is not a global solution to \eqref{eq:jml} as long as \eqref{eq:non-convex} holds.

Second, we use good starting points when solving the optimization \eqref{eq:jml}. Specifically, under the logistic factor model for binary data, a singular-value-decomposition-based algorithm is proposed by \cite{zhang2020note} that is guaranteed
to give a consistent estimator of the model parameters.
Although this estimator is statistically less efficient than the joint-likelihood-based estimator (thus cannot be directly plugged into the likelihood to construct an information criterion), it can serve as a good starting point when solving the optimization \eqref{eq:jml}. For other models, similar singular-value-decomposition-based algorithms can also be developed.

We also discuss the choice of constraint constant $C$ which needs to be specified when computing the constrained joint maximum likelihood estimator. First of all, we point out that it is standard to impose such a constraint for low-rank matrix estimation under nonlinear models. For example, in the work of \cite{cai2013max} on 1-bit matrix completion, it is required that the max norm (i.e., the maximum value of the absolute values of entries) of underlying   low-rank matrix is smaller than a constant, which plays essentially the same role as the constant $C$ in the current work.  Second, according to our simulation study, the estimation of the model parameters and the performance of the proposed information criteria are not sensitive to the choice of $C$, as long as it is set to be sufficiently large. Given a specific dataset, we suggest to run the estimator under different values of $C$ to check its sensitivity. In practice, we suggest to start with a sufficiently large $C$, followed by a sensitivity analysis to check whether the estimator is sensitive to the current choice of $C$.

\section{Information Criteria based on Marginal Likelihood}\label{app:marginal}

We provide some discussion on the behavior of the maximum marginal likelihood when both $N$ and $J$ grow to infinity. To simplify the discussion, we assume there is no missing value and the dispersion parameter is 1, but this discussion can be generalized to the case when there are missing data and the dispersion parameter needs to be estimated. Consider a model with $K$ factors. The marginal likelihood approach assumes that the factors $F_1$, ..., $F_N$ are i.i.d. samples from a known distribution $h$. Then the marginal likelihood function takes the form
$$m_K(A,D) = \sum_{i=1}^N \log( \int \exp\left(l_i(x, A, D) \right) h(x)dx),$$
where $l_i(x, A, D) = \sum_{j=1}^J \log  g(y_{ij}|A_j, d_j, x, 1)$, $A = (A_j: j = 1, ..., J)$, $D = (d_j: j = 1, ..., J)$, and $x \in \mathbb R^K$. Let $(\hat A, \hat D)$ be the estimator based on the marginal likelihood, i.e., $(\hat A, \hat D) \in \arg\max m_K(A, D)$. Furthermore, let
$$\hat F_i  = \arg\max_{x}~~ l_i(x, \hat A, \hat D) \log(h(x)).$$

Then by the Laplace approximation \citep{huber2004estimation} and under suitable regularity conditions, we should be able to establish %{\color{red} Make this sentence more vague?}
% \begin{equation}\label{eq:laplace}
% \begin{aligned}
% m_K(\hat A,\hat D) &= \sum_{i=1}^N \sum_{j=1}^J \log  g(y_{ij}|\hat A_j, \hat d_j, \hat F_i, 1) +  \sum_{i=1}^N \log(h(\hat F_i))\\
% &+ \frac{NK}{2}\log(2\pi/J) -\frac{1}{2}\sum_{i=1}^N \log\left(det(H(\hat F_i, \hat A,\hat D))\right)+ o_p(N),
% \end{aligned}
% \end{equation}
\begin{equation}\label{eq:laplace}
\begin{aligned}
m_K(\hat A,\hat D) &= \sum_{i=1}^N \sum_{j=1}^J \log  g(y_{ij}|\hat A_j, \hat d_j, \hat F_i, 1) +  \sum_{i=1}^N \log(h(\hat F_i))\\
&+ \frac{NK}{2}\log(2\pi/J) -\frac{1}{2}\sum_{i=1}^N \log\left(det(H(\hat F_i, \hat A,\hat D))\right)+ R_{N,J},
\end{aligned}
\end{equation}
where $H(\hat F_i, \hat A,\hat D)$ is the Hessian matrix of $L_i(x) = l_i(x, \hat A, \hat D)$ evaluated at $\hat F_i$ and the $R_{N,J}$ term comes from the remainder term of Laplace approximation. Note that the first term in \eqref{eq:laplace} is the dominant term that takes the same form as the joint likelihood, though $\hat A_j, \hat d_j$, and $\hat F_i$ are obtained from the marginal likelihood. {The remainder $R_{N,J}$ is a term with a smaller asymptotic order.}
 Moreover, we believe that the error bound  established in Theorem~\ref{thm:error-bound-simple} can {be extended} to $\hat M_K = (\hat d_j +  \hat F_i^T \hat A_j)_{N\times J}$ when $K^* \leq K\leq K_{max}$. Therefore, the development in this article will also be useful when developing marginal-likelihood-based information criteria for generalized latent factor models under a high-dimensional setting.

\section{Comparison with Some Related Works}
As discussed in Remark~\ref{rmk:2}, the error bound \eqref{eq:error-bound-simple} improves several recent results on low-rank matrix estimation and completion. We now summarize the comparison in Remark~\ref{rmk:2} using Table~\ref{tab:comparison} below. This comparison focuses on the error bound  \eqref{eq:error-bound-simple} when $K_{max} = K^*$ and data entries are binary and are uniformly missing.

\begin{table}[h]
\centering
\footnotesize
\begin{tabular}{lcccccc}
  \hline
    & Key setting on $M$ &$K^*$  & Error bound \\
    \hline
  Current & $\text{Rank}(M) \leq K^*$ &Can diverge  &$O_p\big[\Big\{\frac{K^* (N\vee J)}{n^*}\Big\}^{1/2}\big]$ \\
  \cite{chen2019structured} & $\text{Rank}(M) \leq K^*$ &Fixed& $O_p\big[\big\{\frac{(N\vee J)}{n^*} + \frac{NJ}{(n^*)^{3/2} }\big\}^{1/2}\big]$ \\
  \cite{bhaskar20151} &$\text{Rank}(M) \leq K^*$ &Can diverge & $O_p\big[ \frac{K^*(N\vee J)^{1/2}}{(n^*)^{1/2}}+ \frac{(N\vee J)^{3}(N\wedge J)^{1/2} (K^*)^{3/2}}{(n^*)^{2}} \big]$ \\
  \cite{ni2016optimal}  &$\text{Rank}(M) \leq K^*$ &Can diverge  & $O_p\big[\big\{\frac{K^*(N\vee J)\log(N+J)}{n^*}\big\}^{1/2}\big]$ \\
   \cite{cai2013max} &  $\Vert M\Vert_* \leq \alpha\sqrt{K^* N J}$ & Can diverge &   $O_p\big[\big\{ \frac{K^*(N\vee J)}{n^*}\big\}^{1/4}\big]$      \\
   \cite{davenport20141}&   $\Vert M\Vert_* \leq \alpha\sqrt{K^* N J}$ & Can diverge &   $O_p\big[\big\{ \frac{K^*(N\vee J)}{n^*}\big\}^{1/4}\big]$      \\
   % \cite{chatterjee2015matrix}&$\Vert M\Vert_* \leq \alpha\sqrt{K^* N J}$ & Can diverge &   $O_p\big[\big\{ \frac{K^*(N\vee J)}{n^*}\big\}^{1/4}\big]$      \\
  \hline
\end{tabular}
\caption{Comparison with existing results on the recovery of $M$. Here, $\Vert \cdot \Vert_*$ denotes the matrix nuclear norm and $\alpha$ is a positive constant.}\label{tab:comparison}
\end{table}

We further compare the current development with \cite{chen2019joint} and \cite{chen2019structured} that also concern likelihood-based analysis of generalized latent factor models. We discuss the similarities and differences below.

\begin{enumerate}
  \item Model: \cite{chen2019structured} and the current work consider the same generalized latent factor model \eqref{eq:models} and \cite{chen2019joint} consider the special case  for binary data as given in Example~\ref{emp:binary}.
  \item Estimation versus selection: The current work establishes results on both the estimation of generalized latent factor model and information criteria for the selection of factors. In contrast, \cite{chen2019joint} and \cite{chen2019structured} focus on the estimation problem.
  \item Confirmatory versus exploratory setting: \cite{chen2019joint} and the current work consider an exploratory factor analysis setting, for which no prior knowledge is assumed on the factor structure. \cite{chen2019structured} focus on a confirmatory factor analysis setting though its results are also generally applicable under an exploratory setting.
  \item Setting on missingness: The current work considers a flexible setting for the missingness of data entries that allows the entries to be non-uniformly missing. In contrast, \cite{chen2019joint} and \cite{chen2019structured} consider a uniformly missing setting which can be viewed as a special case of the current setting.
  \item Optimality: Both the current work and \cite{chen2019structured} establish minimax optimality results on the estimation of generalized latent factor models. The current optimality result, which is established under a more general setting, can be viewed as an extension of that of \cite{chen2019structured}. Minimax optimality is not considered in \cite{chen2019joint}.
\end{enumerate}

In summary, the new contribution of the current paper is of twofold. First, we propose information criteria for selecting the number of factors in high-dimensional generalized latent factor models and establish conditions under which selection consistency is guaranteed. Second, we substantially extend the results on the estimation of
generalized latent factor models under a general setting where the data entries can be non-uniformly missing and the number of factors can also grow to infinity.

\section{Additional Simulation Results}\label{app:addsim}

\subsection{Additional Results for Simulation in Section~4.1}

The average running time for  one independent replication for each of the 48 simulation settings is given in Table~\ref{tab:simtime}, where the computation is run on a computer with an Intel(R) Xeon(R) CPU \@ 2.30GHz.
The computation code for our simulations can be found on the author's Github page: \url{https://github.com/yunxiaochen/JML_IC}.

\begin{table}
\centering
\begin{tabular}{r|ccc|ccc|ccc|cccccccccccc}
  \hline
  & \multicolumn{6}{c|}{$N = J$} & \multicolumn{6}{c}{$N = 5J$} \\
  \hline
 & \multicolumn{3}{c|}{S1} & \multicolumn{3}{c|}{S2} & \multicolumn{3}{c|}{S1} & \multicolumn{3}{c}{S2}\\
  \hline
Average time &M1 & M2 & M3 &M1 & M2 & M3 &M1 & M2 & M3 &M1 & M2 & M3 \\
  \hline
  $J = 100$&6&    9 &  11&   5   & 9&   11&     36 &  48&   73             &33 &  44&   67    \\
  $J = 200$& 21  & 25 &  37&  17   &22&   36&  217 & 195&  314                 &201& 175&  296 \\
  $J = 300$&   53 &  58&   86& 46   &49&   78&   509&  522&  785              &483& 486&  714 \\
  $J = 400$&106  &108 & 161& 92  & 97&  142&    966& 1144& 1423             &869& 1131& 1329 \\
  \hline
\end{tabular}
\caption{The average computation time (in seconds) for running one independent replication for each of the 48 simulation settings.}
\label{tab:simtime}
\end{table}

\subsection{Simulation under Poisson Factor Model}\label{subsec:poisson}
We further provide a simulation study under the Poisson factor model as given in Example~\ref{emp:poisson}.
Similar to the simulation study in Section~4.1, we consider the same factor strength settings S1 and S2, and the
same missing data settings M1-M3. Again, we consider two relationships between $N$ and $J$, including $N = J$ and $N = 5J$. We consider $J = 100, 200, 300, 400$. Again, we let $K^* = 3$ and the true model parameters be generated the similarly as the simulation study in Section~4.1. More precisely, under the setting S1,
the true parameters $d_j^*$, $a_{j1}^*$, ..., $a_{j3}^*$ are generated by sampling independently from the uniform distribution over the interval $[-1,1]$ and the true factor values are generated $f_{i1}^*$, ..., $f_{i3}^*$ are
generated by sampling independently from the uniform distribution over the interval $[-1,1]$. Under the setting S2, $f_{i3}^*$ is generated from the uniform distribution over the interval $[-0.4,0.4]$ and the rest of the parameters are generated the same as those in S1. We use the proposed JIC to select $K$ from the candidate set $\{1,2,3,4, 5\}$ and the constraint constant $C$ in \eqref{eq:jml} is set to be 3. The true model parameters satisfy this constraint. There are 48 simulation settings in total and 100 independent replications are run for each setting.
Figure~\ref{fig:merr} below shows the value of
$\max_{3\leq K\leq 5} \{(NJ)^{-1/2} \|\hat{M}^{(K)}-M^* \|_F \}$ under different settings.   Table~\ref{tab:tabF1}
shows the accuracy on determining the number of factors. Finally, Table~\ref{tab:simtime2} gives the average
 average running time for  one independent replication for each of the 48 simulation settings.

\begin{figure}[h]
  \centering
  \begin{subfigure}[b]{0.3\textwidth}
         \centering
         \includegraphics[width=\textwidth]{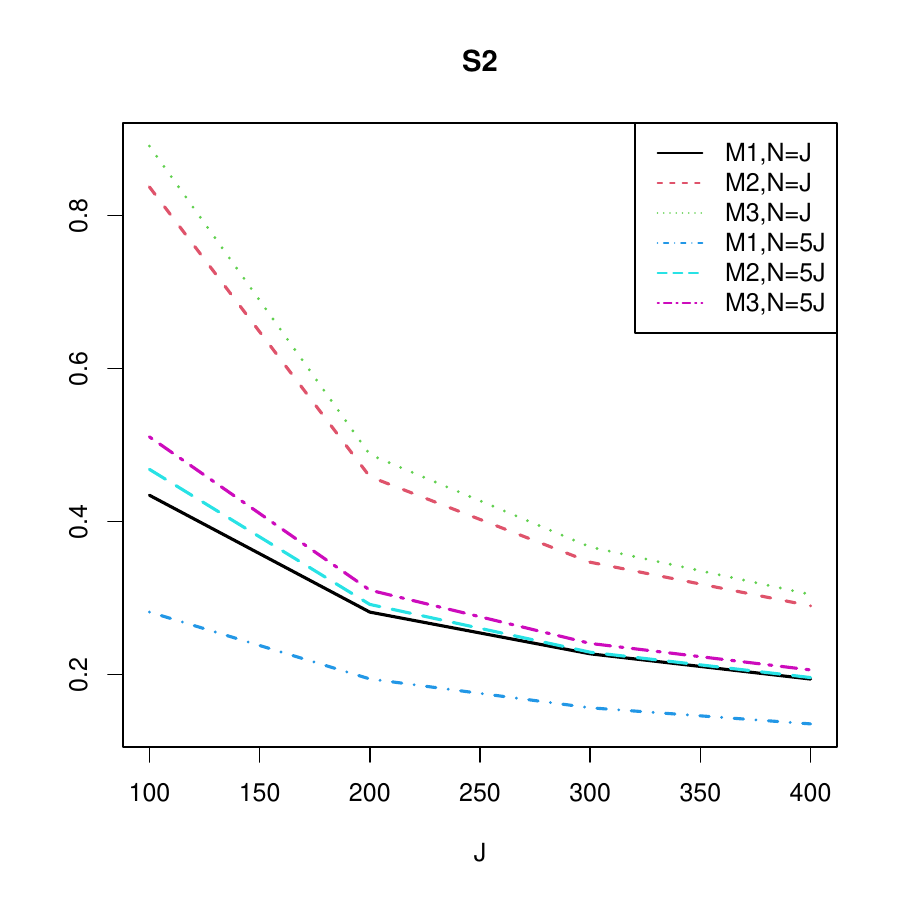}
         \label{fig:y equals x}
     \end{subfigure}
     \begin{subfigure}[b]{0.3\textwidth}
         \centering
         \includegraphics[width=\textwidth]{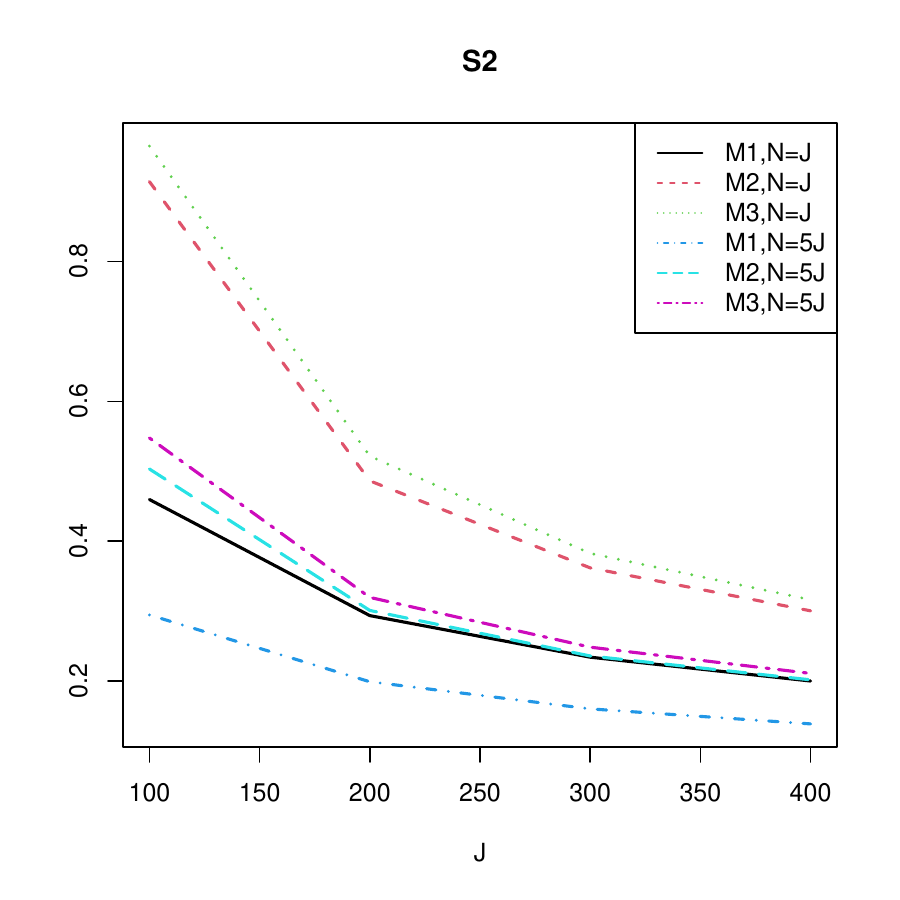}
         \label{fig:three sin x}
     \end{subfigure}
  \caption{The loss $\max_{3\leq K\leq 5} \{(NJ)^{-1/2} \|\hat{M}^{(K)}-M^* \|_F \}$ for the recovery of the low-rank matrix $M^*$, where each point is the mean loss calculated by averaging over 100 independent replications. Panels (a) and (b) show the results under the two different factor strength settings, S1 and S2, respectively.}\label{fig:merr}
\end{figure}

\begin{table}
\centering
\begin{tabular}{r|ccc|ccc|ccc|cccccccccccc}
  \hline
  & \multicolumn{6}{c|}{$N = J$} & \multicolumn{6}{c}{$N = 5J$} \\
  \hline
 & \multicolumn{3}{c|}{S1} & \multicolumn{3}{c|}{S2} & \multicolumn{3}{c|}{S1} & \multicolumn{3}{c}{S2}\\
  \hline
 Under-selection &M1 & M2 & M3 &M1 & M2 & M3 &M1 & M2 & M3 &M1 & M2 & M3 \\
  \hline
  $J = 100$& 0&0&0& 58 & 40&  33& 0&0&0&100& 100& 100 \\
  $J = 200$& 0&0&0& 3  &47 & 49 &0&0&0&12& 100& 100 \\
  $J = 300$& 0&0&0&0   &2  & 3& 0&0&0& 0&  87&  83\\
  $J = 400$&0&0&0& 0   &0  & 0 & 0&0&0& 0&   2&   2 \\
  \hline
   Over-selection&  \\
   \hline
  $J = 100$& 0&19&13& 0&19&9& 0&0&0&0&0&0 \\
  $J = 200$& 0&0&0& 0&0&0& 0&0&0&0&0&0 \\
  $J = 300$& 0&0&0& 0&0&0& 0&0&0& 0&0&0\\
  $J = 400$& 0&0&0&  0&0&0& 0&0&0&0&0&0 \\
  \hline
\end{tabular}
\caption{The number of times that the true number of factors is under- or over-selected selected among 100 independent replications under each of the 48 simulation settings.}
\label{tab:tabF1}
\end{table}

           \begin{table}
\centering
\begin{tabular}{r|ccc|ccc|ccc|cccccccccccc}
  \hline
  & \multicolumn{6}{c|}{$N = J$} & \multicolumn{6}{c}{$N = 5J$} \\
  \hline
 & \multicolumn{3}{c|}{S1} & \multicolumn{3}{c|}{S2} & \multicolumn{3}{c|}{S1} & \multicolumn{3}{c}{S2}\\
  \hline
Average time &M1 & M2 & M3 &M1 & M2 & M3 &M1 & M2 & M3 &M1 & M2 & M3 \\
  \hline
  $J = 100$&  1 &  1  & 1 &  1&   1&   1&      12  & 6&   7 & 11&   6&   7 \\
  $J = 200$&  10&   5 &  5&    9  & 5 &  5 &    79&  40&  43 & 73 & 35&  39  \\
  $J = 300$&  28&  14 & 16&  25  &13 & 14 &   249 &122 &125 &  222& 107& 110    \\
  $J = 400$&  60&  30 & 32&   53  &28&  29&   536& 267 &278 &475 &229& 242\\
  \hline
\end{tabular}
\caption{The average computation time (in seconds) for running one independent replication for each of the 48 simulation settings.}
\label{tab:simtime2}
\end{table}

\subsection{A Scree Plot Example}

Scree plots are a widely used tool for selecting the number of factors in factor analysis \citep{cattell1966scree}. A scree plot displays the eigenvalues of the covariance matrix of data in a downward curve, ordering the eigenvalues from largest to smallest. The number of factors is then determined by finding the ``elbow" of the graph. The ``elbow" is the eigenvalue where the eigenvalues seem to level off and the number of factors is determined by the number of eigenvalues that are greater than the elbow. This approach typically works well for data following a linear factor model. This is because, under a linear factor model, the covariance matrix of data is approximately a low-rank matrix plus a diagonal matrix, where the low-rank part drives the ``elbow" phenomenon. When data are generated from a nonlinear factor model, such as the logistic or Poisson factor models considered in the current work, the factor structure of data cannot be fully characterized by the covariance matrix. In particular, the covariance matrix cannot be approximated by a low-rank matrix plus a diagonal matrix. As a result, the elbow of the scree plot may no longer correspond to the number of factors. We provide a simulated example to illustrate this point. Figure~\ref{fig:figF1} shows the scree plot for data generated from a Poisson factor model under a setting when $N = J = 200$, $K^* = 3$, and there are no missing entries.
Based on the scree plot, one may tend to choose seven or eight factors, which is larger than the true number of factors. %Note that according to the simulation results in Appendix~\ref{subsec:poisson}, the proposed JIC always correctly chooses the true number of factors under this setting.
%{\color{red}(XL: it looks like 3 factors to me. Do we have a more obvious example?)}

\begin{figure}[h]
  \centering
  \includegraphics[scale = 0.5]{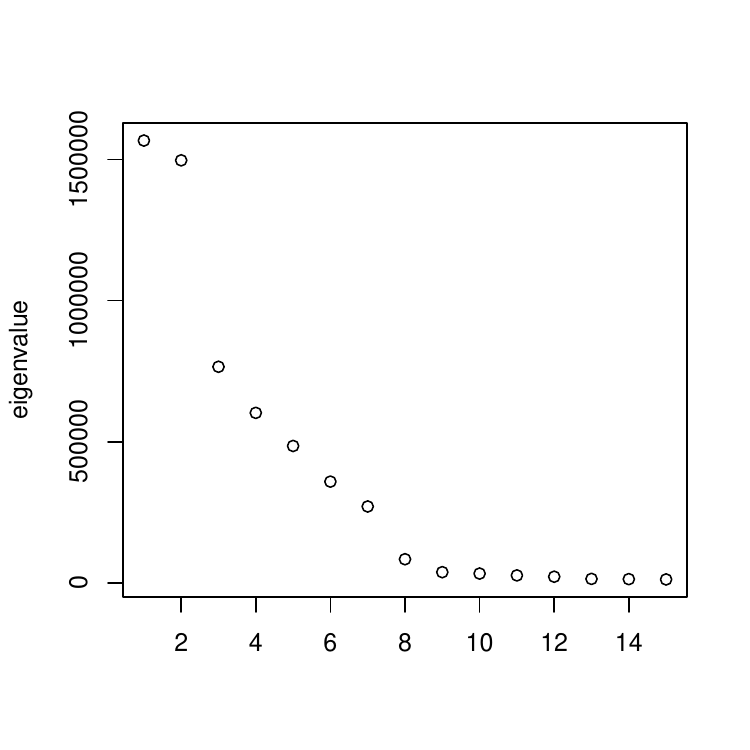}
  \caption{The scree plot for data that are generated from a Poisson factor model with three factors}\label{fig:figF1}
\end{figure}

\subsection{Comparison with \cite{bai2018consistency}}

We now compare the proposed JIC with the method proposed in \cite{bai2018consistency} via a simulation study. In this study, data are generated from a linear factor model, with $N = 2J$. More specifically, $y_{ij}$ follows a normal distribution with mean $d_j + A_j^T F_i$ and variance 1, so that the data follow a spike covariance structure model assumed in \cite{bai2018consistency}, with the eigenvalues of the covariance matrix of $(y_{i1}, ..., y_{iJ})^T$ satisfying $\lambda_{K^*} > \lambda_{K^*+1} = \cdots =\lambda_{J} = 1$. Again, we let $K^* = 3$ and the true model parameters be generated under the setting S1  the simulation study in Section~4.1, where the three factors are of the same strength. More precisely,
the true parameters $d_j^*$, $a_{j1}^*$, ..., $a_{j3}^*$ are generated by sampling independently from the uniform distribution over the interval $[-2,2]$ and the true factor values are generated $f_{i1}^*$, ..., $f_{i3}^*$ are
generated by sampling independently from the uniform distribution over the interval $[-2,2]$. We consider $J = 10, 20, ..., 50$ and $N = 2J$. Note that under this linear factor model with no missing data and assuming that the variance of $y_{ij}$ is known to be 1, then the proposed JIC is the same as the $PC_{p3}$ criterion proposed in \cite{bai2002determining}.

We use the proposed JIC to select $K$ from the candidate set $\{1,2,3,4, 5\}$ and the constraint constant $C$ in \eqref{eq:jml} is set to be 5. In addition, we also use the AIC and BIC proposed in \cite{bai2018consistency} to select $K$ from the candidate set $\{1,2,3,4, 5\}$. The results are given in Table~\ref{tab:tabF2}. As we can see, all three information criteria become more accurate when $N$ and $J$ simultaneously grow. Specifically, the proposed JIC and the BIC in \cite{bai2018consistency} perform similarly. When $J \geq 20$ and $N = 2J$, both methods correctly identify the true number of factors all the time. When $J = 10$ and $N = 20$, both the JIC and the BIC are  correct 92\% of the times, though the two methods are slightly different in the numbers of over- and under-selections. Finally, consistent with the observations in \cite{bai2018consistency}, the AIC is less accurate and tends to over-select.

\begin{table}
\centering
\begin{tabular}{r|ccc|ccc}
  \hline
  & \multicolumn{3}{c|}{Under-selection} &\multicolumn{3}{c}{Over-selection}   \\
  \hline
 &JIC & BIC & AIC &JIC & BIC & AIC \\
  \hline
  $J = 10$&2&4&2&6&4&31\\
  $J = 20$&0&0&0&0&0&18\\
  $J = 30$&0&0&0&0&0&7 \\
  $J = 40$&0&0&0&0&0&3 \\
  $J = 50$&0&0&0&0&0&2 \\
  \hline
\end{tabular}
\caption{The number of times that the true number of factors is under- or over-selected selected among 100 independent replications under each of the 5 simulation settings, under a linear factor model with a spike covariance structure.}
\label{tab:tabF2}
\end{table}

\section{Additional Results for Real Data Analysis}\label{app:real}

In what follows, we provide additional results for the real data analysis. In Tables~\ref{tab:3factor_full} and \ref{tab:corr_3factor}, we show the loading matrix and the sample covariance matrix for the estimated factor scores, after applying the oblimin rotation. Note that the items have been reordered, with items 1-32, 33-55, and 56-79 designed to measure the  psychoticism, extraversion, and neuroticism traits, respectively. The content of the items can be found in \cite{eysenck1985revised}. Note that our data have been pre-processed so that the negatively worded items are reversely scored. As we can see, items 1-32, 33-55, and 56-79 tend to have high loadings on F2, F1, and F3, respectively. According to Table \ref{tab:corr_3factor}, the correlations between the three estimated factors are relatively small, suggesting that the three factors tend to be uncorrelated.

\begin{table}[h]
  \centering
  \footnotesize
  \begin{tabular}{r|ccc}
    \hline
   \multicolumn{1}{r|}{Item } & \multicolumn{1}{c}{F1} & \multicolumn{1}{c}{F2} & \multicolumn{1}{c}{F3}\\
\hline
\textbf{1 }      &	0.31&	2.33&	0.42	\\
\textbf{2 }      &	0.34&	1.37&	-0.11	\\
\textbf{3 }      &	0.53&	1.18&	0.49	\\
\textbf{4 }      &	0.27&	1.47&	0.79	\\
\textbf{5 }      &	0.89&	1.37&	0.03	\\
\textbf{6 }      &	0.44&	1.11&	0.23	\\
\textbf{7 }      &	-0.25&	1.99&	0.04	\\
\textbf{8 }      &	0.35&	0.83&	-0.23	\\
\textbf{9 }      &	-0.58&	1.16&	0.50	\\
\textbf{10}      &	-0.04&	1.59&	0.71	\\
\textbf{11}      &	0.22&	0.85&	-0.10	\\
\textbf{12}      &	0.03&	1.78&	0.36	\\
\textbf{13}      &	0.03&	0.45&	0.50	\\
\textbf{14}      &	0.92&	0.95&	0.27	\\
\textbf{15}      &	-0.15&	1.04&	-0.97	\\
\textbf{16}      &	0.55&	1.13&	-0.53	\\
\textbf{17}      &	0.08&	0.63&	-0.01	\\
\textbf{18}      &	-0.06&	0.93&	-0.35	\\
\textbf{19}      &	0.13&	0.58&	-0.31	\\
\textbf{20}      &	0.08&	1.78&	-0.22	\\
\textbf{21}      &	-0.50&	2.37&	-0.63	\\
\textbf{22}      &	-0.49&	2.17&	-0.64	\\
\textbf{23}      &	-0.54&	1.55&	0.02	\\
\textbf{24}      &	0.23&	1.15&	-0.47	\\
\textbf{25}      &	0.18&	0.77&	-0.06	\\
\textbf{26}      &	-0.35&	1.15&	0.10	\\
\textbf{27}      &	0.44&	1.85&	0.13	\\
\textbf{28}      &	0.95&	1.02&	0.38	\\
\textbf{29}      &	-0.16&	0.50&	0.45	\\
\textbf{30}      &	0.16&	1.31&	-0.23	\\
\textbf{31}      &	-0.08&	1.25&	-0.05	\\
\textbf{32}      &	-0.18&	0.58&	-0.25	\\
\hline				
\textbf{33}      &	0.33&	-0.17&	-0.24	\\
\textbf{34}      &	2.75&	-0.19&	0.48	\\
\textbf{35}      &	3.61&	-0.31&	-0.08	\\
\textbf{36}      &	2.05&	0.08&	-0.10	\\
\textbf{37}      &	2.08&	-0.34&	-0.41	\\
\textbf{38}      &	1.59&	0.03&	0.03	\\
\textbf{39}      &	1.97&	-0.77&	-0.44	\\
\textbf{40}      &	1.00&	0.18&	-0.58	\\
\hline
\end{tabular}
\quad
  \begin{tabular}{c|ccc}
    \hline
   \multicolumn{1}{c|}{Item } & \multicolumn{1}{c}{F1} & \multicolumn{1}{c}{F2} & \multicolumn{1}{c}{F3} \\
\hline
\textbf{41}       &	1.84&	-0.23&	-0.09	\\
\textbf{42}       &	2.98&	0.36&	-0.07	\\
\textbf{43}       &	0.91&	-0.07&	-0.05	\\
\textbf{44}       &	2.59&	-0.98&	0.15	\\
\textbf{45}       &	1.19&	0.96&	0.65	\\
\textbf{46}       &	0.49&	-0.02&	-0.10	\\
\textbf{47}       &	0.79&	0.36&	-0.33	\\
\textbf{48}       &	0.93&	0.59&	0.18	\\
\textbf{49}       &	0.43&	-0.02&	0.11	\\
\textbf{50}       &	2.59&	-0.01&	-0.12	\\
\textbf{51}       &	1.92&	-0.12&	0.00	\\
\textbf{52}       &	3.78&	-0.02&	0.10	\\
\textbf{53}       &	3.79&	0.54&	-0.16	\\
\textbf{54}       &	1.81&	-0.18&	-0.01	\\
\textbf{55}       &	2.73&	0.08&	-0.08	\\
\hline				
\textbf{56}       &	0.34&	0.73&	2.29	\\
\textbf{57}       &	0.13&	0.41&	1.58	\\
\textbf{58}       &	0.37&	-1.10&	2.13	\\
\textbf{59}       &	0.01&	0.67&	1.64	\\
\textbf{60}       &	-0.01&	-0.18&	1.78	\\
\textbf{61}       &	0.01&	0.45&	2.10	\\
\textbf{62}       &	0.39&	-0.02&	1.68	\\
\textbf{63}       &	-0.47&	0.11&	2.10	\\
\textbf{64}       &	-0.35&	-0.54&	2.84	\\
\textbf{65}       &	-0.09&	-0.18&	1.38	\\
\textbf{66}       &	-0.23&	0.49&	1.91	\\
\textbf{67}       &	0.13&	-0.07&	0.88	\\
\textbf{68}       &	0.04&	0.34&	0.59	\\
\textbf{69}       &	0.16&	0.40&	1.25	\\
\textbf{70}       &	0.04&	0.61&	1.36	\\
\textbf{71}       &	0.71&	-0.16&	1.17	\\
\textbf{72}       &	-0.11&	0.73&	0.77	\\
\textbf{73}       &	-0.28&	-0.58&	2.25	\\
\textbf{74}       &	-0.23&	0.30&	2.02	\\
\textbf{75}       &	-0.17&	0.70&	1.55	\\
\textbf{76}       &	-0.26&	-0.42&	1.81	\\
\textbf{77}       &	0.85&	0.38&	1.48	\\
\textbf{78}       &	0.31&	0.08&	1.32	\\
\textbf{79}       &	0.45&	0.53&	1.00	\\
&&&\\
\hline
\end{tabular}
\caption{Estimated loading matrix for the three-factor model after applying the oblimin rotation. }\label{tab:3factor_full}
\end{table}

\begin{table}[h]
  \centering
  \begin{tabular}{cccc}
    \hline
  &F1 &F2 & F3\\
  \hline
  F1 &  1.00 &-0.03& -0.19\\
  F2 &  -0.03& 1.00& -0.02\\
  F3 & -0.19&-0.02&  1.00\\
    \hline
  \end{tabular}
  \caption{The sample covariance matrix for the estimated factor scores, under the three-factor model after applying the oblimin rotation. Note that the model parameters have been rescaled, so that the sample variance for each factor is one. }\label{tab:corr_3factor}
\end{table}

We further provide results for the two- and four-factor models, whose JIC values are also relatively small. These results may provide us further insights about the latent structure of this personality inventory. Tables~\ref{tab:2factor_full} and \ref{tab:4factor_full} provide the the loading matrices for the two models, respectively,
after applying the oblimin rotation. Moreover, Tables \ref{tab:corr_2factor} and \ref{tab:corr_4factor}
show the sample covariance matrices for the estimated factor scores, from the two models, respectively.
According to Table~\ref{tab:2factor_full}, the items that are designed to measure the extraversion trait tend to
have high loadings for the first factor and items designed to measure neuroticism tend to have high loadings for the second factor, while most items designed to measure psychoticism have small loadings for both factors.
These results suggest that the psychoticism factor may not be captured by the two-factor model.

From the loading structure given in Table~\ref{tab:4factor_full}, the extracted factors F4, F2, and F3 tend to correspond to the psychoticism, extraversion, and neuroticism traits, respectively. In addition, most items have small loadings on F1, except for items 14. ``Do you stop to think things over before doing anything?", 28. ``Do you generally ‘look before you leap’?", 45. ``Have people said that you sometimes act too rashly?", and 48. ``Do you often make decisions on the spur of the moment?", where items 14 and 28 are negatively worded and thus reversely scored. It seems a minor factor about impulsive decision.

 %suggesting that the psychoticism factor is not captured by the two-factor model.

\begin{table}[h]
  \centering
  \footnotesize
  \begin{tabular}{r|ccc}
    \hline
   \multicolumn{1}{r|}{Item } & \multicolumn{1}{c}{F1} & \multicolumn{1}{c}{F2} \\
\hline
\textbf{1 }      &	0.59&	0.80	\\
\textbf{2 }      &	0.46&	0.17	\\
\textbf{3 }      &	0.72&	0.78	\\
\textbf{4 }      &	0.48&	1.01	\\
\textbf{5 }      &	1.04&	0.38	\\
\textbf{6 }      &	0.55&	0.39	\\
\textbf{7 }      &	0.16&	0.40	\\
\textbf{8 }      &	0.45&	0.01	\\
\textbf{9 }      &	-0.36&	0.68	\\
\textbf{10}      &	0.22&	0.89	\\
\textbf{11}      &	0.38&	0.11	\\
\textbf{12}      &	0.45&	0.68	\\
\textbf{13}      &	0.19&	0.64	\\
\textbf{14}      &	0.98&	0.48	\\
\textbf{15}      &	0.11&	-0.53	\\
\textbf{16}      &	0.73&	-0.20	\\
\textbf{17}      &	0.25&	0.18	\\
\textbf{18}      &	0.19&	-0.10	\\
\textbf{19}      &	0.27&	-0.11	\\
\textbf{20}      &	0.32&	0.15	\\
\textbf{21}      &	0.11&	0.04	\\
\textbf{22}      &	0.12&	0.02	\\
\textbf{23}      &	-0.14&	0.39	\\
\textbf{24}      &	0.43&	-0.13	\\
\textbf{25}      &	0.35&	0.12	\\
\textbf{26}      &	-0.11&	0.32	\\
\textbf{27}      &	0.71&	0.55	\\
\textbf{28}      &	1.06&	0.63	\\
\textbf{29}      &	-0.02&	0.54	\\
\textbf{30}      &	0.35&	0.09	\\
\textbf{31}      &	0.14&	0.21	\\
\textbf{32}      &	-0.03&	-0.09	\\
\hline			
\textbf{33}      &	0.24&	-0.28	\\
\textbf{34}      &	2.56&	0.39	\\
\textbf{35}      &	3.29&	-0.14	\\
\textbf{36}      &	2.12&	-0.07	\\
\textbf{37}      &	1.98&	-0.50	\\
\textbf{38}      &	1.60&	0.06	\\
\textbf{39}      &	1.62&	-0.59	\\
\textbf{40}      &	1.05&	-0.48	\\
\hline
\end{tabular}
\quad
  \begin{tabular}{c|ccc}
    \hline
   \multicolumn{1}{c|}{Item } & \multicolumn{1}{c}{F1} & \multicolumn{1}{c}{F2}   \\
\hline
\textbf{41}       &	1.69&	-0.13	\\
\textbf{42}       &	2.54&	-0.01	\\
\textbf{43}       &	0.85&	-0.08	\\
\textbf{44}       &	2.18&	-0.14	\\
\textbf{45}       &	1.24&	0.81	\\
\textbf{46}       &	0.46&	-0.08	\\
\textbf{47}       &	0.86&	-0.22	\\
\textbf{48}       &	1.04&	0.31	\\
\textbf{49}       &	0.40&	0.10	\\
\textbf{50}       &	2.38&	-0.10	\\
\textbf{51}       &	1.88&	-0.02	\\
\textbf{52}       &	3.51&	0.10	\\
\textbf{53}       &	3.79&	0.03	\\
\textbf{54}       &	1.74&	-0.06	\\
\textbf{55}       &	2.85&	-0.05	\\
\hline			
\textbf{56}       &	0.43&	2.38	\\
\textbf{57}       &	0.15&	1.62	\\
\textbf{58}       &	-0.09&	1.33	\\
\textbf{59}       &	0.12&	1.69	\\
\textbf{60}       &	-0.14&	1.54	\\
\textbf{61}       &	0.08&	2.10	\\
\textbf{62}       &	0.24&	1.46	\\
\textbf{63}       &	-0.55&	1.93	\\
\textbf{64}       &	-0.54&	2.07	\\
\textbf{65}       &	-0.16&	1.17	\\
\textbf{66}       &	-0.18&	1.99	\\
\textbf{67}       &	0.06&	0.79	\\
\textbf{68}       &	0.10&	0.63	\\
\textbf{69}       &	0.22&	1.34	\\
\textbf{70}       &	0.09&	1.44	\\
\textbf{71}       &	0.53&	0.99	\\
\textbf{72}       &	0.02&	0.88	\\
\textbf{73}       &	-0.48&	1.65	\\
\textbf{74}       &	-0.24&	1.97	\\
\textbf{75}       &	-0.07&	1.59	\\
\textbf{76}       &	-0.40&	1.48	\\
\textbf{77}       &	0.86&	1.52	\\
\textbf{78}       &	0.22&	1.26	\\
\textbf{79}       &	0.51&	1.13	\\
&&&\\
\hline
\end{tabular}
\caption{Estimated loading matrix for the two-factor model after applying the oblimin rotation. }\label{tab:2factor_full}
\end{table}

\begin{table}[h]
  \centering
  \begin{tabular}{cccc}
    \hline
  &F1 &F2  \\
  \hline
  F1 &  1.00 &-0.22 \\
  F2 &  -0.22& 1.00 \\
    \hline
  \end{tabular}
  \caption{The sample covariance matrix for the estimated factor scores, under the two-factor model after applying the oblimin rotation.   }\label{tab:corr_2factor}
\end{table}

\begin{table}[h]
  \centering
  \footnotesize
  \begin{tabular}{r|cccc}
    \hline
   \multicolumn{1}{r|}{Item } & \multicolumn{1}{c}{F1} & \multicolumn{1}{c}{F2} & \multicolumn{1}{c}{F3} & \multicolumn{1}{c}{F4} \\
\hline
\textbf{1 }      &		0.48&	0.39&	0.48&	2.31	\\
\textbf{2 }      &		0.50&	0.25&	-0.09&	1.31	\\
\textbf{3 }      &		0.31&	0.56&	0.54&	1.43	\\
\textbf{4 }      &		0.22&	0.34&	0.86&	1.49	\\
\textbf{5 }      &		0.67&	0.74&	0.02&	1.17	\\
\textbf{6 }      &		0.53&	0.25&	0.20&	0.92	\\
\textbf{7 }      &		-0.19&	0.12&	0.17&	2.68	\\
\textbf{8 }      &		0.06&	0.41&	-0.20&	0.96	\\
\textbf{9 }      &		0.00&	-0.45&	0.53&	1.36	\\
\textbf{10}      &		0.07&	0.08&	0.81&	1.89	\\
\textbf{11}      &		0.37&	0.08&	-0.12&	0.73	\\
\textbf{12}      &		-0.12&	0.30&	0.46&	2.36	\\
\textbf{13}      &		0.24&	-0.01&	0.50&	0.38	\\
\textbf{14}      &		4.87&	-0.27&	-0.08&	0.32	\\
\textbf{15}      &		0.30&	-0.26&	-1.00&	1.03	\\
\textbf{16}      &		0.98&	0.15&	-0.67&	0.73	\\
\textbf{17}      &		0.43&	-0.07&	-0.03&	0.53	\\
\textbf{18}      &		0.24&	-0.05&	-0.47&	1.23	\\
\textbf{19}      &		-0.06&	0.24&	-0.27&	0.69	\\
\textbf{20}      &		0.07&	0.22&	-0.17&	2.06	\\
\textbf{21}      &		-0.47&	0.07&	-0.44&	3.15	\\
\textbf{22}      &		-0.38&	-0.03&	-0.51&	3.08	\\
\textbf{23}      &		0.30&	-0.52&	0.08&	1.55	\\
\textbf{24}      &		0.23&	0.21&	-0.46&	1.18	\\
\textbf{25}      &		0.12&	0.15&	-0.06&	0.78	\\
\textbf{26}      &		-0.07&	-0.24&	0.14&	1.35	\\
\textbf{27}      &		0.84&	0.23&	0.11&	1.51	\\
\textbf{28}      &		5.39&	-0.19&	0.11&	0.35	\\
\textbf{29}      &		0.11&	-0.13&	0.47&	0.48	\\
\textbf{30}      &		0.13&	0.25&	-0.18&	1.52	\\
\textbf{31}      &		-0.07&	0.15&	0.04&	1.55	\\
\textbf{32}      &		-0.18&	-0.01&	-0.20&	0.79	\\
\hline						
\textbf{33}      &		-0.24&	0.46&	-0.20&	-0.12	\\
\textbf{34}      &		0.46&	2.55&	0.45&	-0.44	\\
\textbf{35}      &		0.12&	3.67&	-0.05&	-0.48	\\
\textbf{36}      &		-0.06&	2.19&	-0.05&	0.13	\\
\textbf{37}      &		-0.85&	2.48&	-0.25&	0.00	\\
\textbf{38}      &		-0.07&	1.66&	0.06&	0.04	\\
\textbf{39}      &		-0.27&	2.08&	-0.42&	-0.71	\\
\textbf{40}      &		0.71&	0.71&	-0.71&	-0.13	\\
\hline
\end{tabular}
\quad
  \begin{tabular}{c|cccc}
    \hline
   \multicolumn{1}{r|}{Item } & \multicolumn{1}{c}{F1} & \multicolumn{1}{c}{F2} & \multicolumn{1}{c}{F3} & \multicolumn{1}{c}{F4} \\
\hline
\textbf{41}       &	-0.19&	1.97&	-0.06&	-0.11	\\
\textbf{42}       &	-0.29&	3.96&	0.03&	0.72	\\
\textbf{43}       &	-0.20&	1.05&	0.00&	0.00	\\
\textbf{44}       &	-0.48&	2.76&	0.22&	-0.68	\\
\textbf{45}       &	2.25&	0.60&	0.58&	0.32	\\
\textbf{46}       &	0.10&	0.48&	-0.10&	-0.04	\\
\textbf{47}       &	0.61&	0.57&	-0.39&	0.16	\\
\textbf{48}       &	5.43&	0.32&	-0.14&	-0.69	\\
\textbf{49}       &	0.36&	0.31&	0.09&	-0.25	\\
\textbf{50}       &	-0.32&	3.33&	-0.05&	0.22	\\
\textbf{51}       &	0.35&	1.78&	-0.04&	-0.30	\\
\textbf{52}       &	0.49&	3.55&	0.08&	-0.37	\\
\textbf{53}       &	0.24&	3.96&	-0.12&	0.52	\\
\textbf{54}       &	-0.26&	1.97&	0.04&	-0.06	\\
\textbf{55}       &	0.23&	2.61&	-0.08&	-0.02	\\
\hline					
\textbf{56}       &	0.84&	0.07&	2.20&	0.33	\\
\textbf{57}       &	0.57&	-0.08&	1.50&	0.14	\\
\textbf{58}       &	-0.15&	0.47&	2.14&	-1.22	\\
\textbf{59}       &	0.53&	-0.13&	1.65&	0.43	\\
\textbf{60}       &	0.26&	-0.08&	1.74&	-0.41	\\
\textbf{61}       &	0.42&	-0.09&	2.05&	0.24	\\
\textbf{62}       &	-0.01&	0.44&	1.73&	-0.05	\\
\textbf{63}       &	-0.42&	-0.26&	2.48&	0.35	\\
\textbf{64}       &	-0.65&	-0.02&	3.18&	-0.36	\\
\textbf{65}       &	-0.27&	0.10&	1.52&	-0.05	\\
\textbf{66}       &	-0.05&	-0.12&	2.13&	0.60	\\
\textbf{67}       &	-0.17&	0.25&	0.94&	0.03	\\
\textbf{68}       &	0.23&	-0.02&	0.60&	0.26	\\
\textbf{69}       &	0.61&	-0.10&	1.19&	0.17	\\
\textbf{70}       &	0.56&	-0.16&	1.32&	0.41	\\
\textbf{71}       &	0.11&	0.69&	1.16&	-0.25	\\
\textbf{72}       &	0.11&	-0.09&	0.83&	0.72	\\
\textbf{73}       &	-0.05&	-0.29&	2.25&	-0.63	\\
\textbf{74}       &	-0.29&	-0.01&	2.42&	0.47	\\
\textbf{75}       &	0.20&	-0.17&	1.57&	0.63	\\
\textbf{76}       &	0.32&	-0.40&	1.80&	-0.73	\\
\textbf{77}       &	0.94&	0.46&	1.41&	-0.07	\\
\textbf{78}       &	0.57&	0.08&	1.29&	-0.26	\\
\textbf{79}       &	0.53&	0.29&	0.99&	0.36	\\
&&&\\
\hline
\end{tabular}
\caption{Estimated loading matrix for the four-factor model after applying the oblimin rotation. }\label{tab:4factor_full}
\end{table}

\begin{table}[h]
  \centering
  \begin{tabular}{ccccc}
    \hline
  &F1 &F2 & F3  & F4  \\
  \hline
F1& 1.00  &0.22 &-0.07 & 0.07\\
F2& 0.22  &1.00 &-0.20 &-0.03\\
F3& -0.07 &-0.20&  1.00&  0.01\\
F4& 0.07  &-0.03&  0.01&  1.00\\
    \hline
  \end{tabular}
  \caption{The sample covariance matrix for the estimated factor scores, under the four-factor model after applying the oblimin rotation.   }\label{tab:corr_4factor}
\end{table}

We compare the proposed method with the classical Akaike information criterion (AIC) and Bayesian information criterion (BIC) calculated based on the marginal likelihood function, where the latent factors are treated as random variables. More specifically, the latent factors are assumed to follow a multivariate normal distribution, in the calculation of the marginal likelihood. The marginal maximum likelihood estimator is computed using the R package ``mirt" \citep{chalmers2012mirt}, where the computation for the marginal maximum likelihood estimator is  carried out using an Expectation-Maximization (EM) algorithm. The EM algorithm is very time-consuming when only involving a moderate number of factors \citep{reckase2009multidimensional}.
The AIC and BIC values for the one- through five-factor models are given in Table~\ref{tab:marglik} below. In calculating the AIC and BIC values,
the number of parameters for a $K$-factor model is
$J(K+1) - K(K-1)/2$, recalling that $J$ is the number of items.
The three-factor model fits best according to the BIC value, which is consistent with the selection based on the proposed JIC. On the other hand, AIC selects the four-factor model.
Note that under the classical asymptotic regime and the true model is one of the candidate models, the BIC guarantees consistency for model selection, while the AIC tends to over-select \citep{shao1997asymptotic}.

\begin{table}[h]
  \centering
  \begin{tabular}{cccccccc}
    \hline
    $K$ & 1 & 2 & 3 & 4 &5 \\
    \hline
    AIC & 66304 & 63434& 61942          &\textbf{61732} &61750\\
    BIC & 67049 & 64547& \textbf{63418} & 63566&63937\\
    \hline
  \end{tabular}
  \caption{AIC and BIC values based on the marginal likelihood for the one- through five-factor models.}\label{tab:marglik}
\end{table}

Finally, we provide the estimation results for the three-factor model from the marginal-likelihood approach, in comparison with those from the joint-likelihood approach. The results are given in Tables~\ref{tab:3factor_full_mar} and \ref{tab:corr_3factor_mar}. Similar to the analysis above, the results are under the oblimin rotation. As we can see, although the estimates are slightly different from those given by the joint likelihood, the loading structure
is similar and suggests that the three factors correspond to the extraversion, psychoticism, and neuroticism traits, respectively.

% obtained from the two methods are similar,

%and the four-factor model fits best according the AIC.
%Therefore, the model-selection result based on the BIC is consistent with that based on the proposed JIC.

\begin{table}[h]
  \centering
  \footnotesize
  \begin{tabular}{r|ccc}
    \hline
   \multicolumn{1}{r|}{Item } & \multicolumn{1}{c}{F1} & \multicolumn{1}{c}{F2} & \multicolumn{1}{c}{F3}\\
\hline
\textbf{1 }      &		0.16&	1.57&	0.37	\\
\textbf{2 }      &		0.20&	1.12&	-0.15	\\
\textbf{3 }      &		0.47&	1.10&	0.57	\\
\textbf{4 }      &		0.14&	1.11&	0.73	\\
\textbf{5 }      &		0.84&	1.08&	0.04	\\
\textbf{6 }      &		0.27&	0.91&	0.15	\\
\textbf{7 }      &		-0.30&	1.38&	0.02	\\
\textbf{8 }      &		0.22&	0.72&	-0.24	\\
\textbf{9 }      &		-0.63&	0.91&	0.48	\\
\textbf{10}      &		-0.20&	1.21&	0.59	\\
\textbf{11}      &		0.17&	0.70&	-0.07	\\
\textbf{12}      &		-0.12&	1.39&	0.30	\\
\textbf{13}      &		0.00&	0.43&	0.47	\\
\textbf{14}      &		0.72&	0.77&	0.23	\\
\textbf{15}      &		-0.14&	0.85&	-0.81	\\
\textbf{16}      &		0.50&	0.94&	-0.48	\\
\textbf{17}      &		0.04&	0.54&	0.00	\\
\textbf{18}      &		-0.26&	1.13&	-0.54	\\
\textbf{19}      &		0.12&	0.48&	-0.29	\\
\textbf{20}      &		0.00&	1.22&	-0.16	\\
\textbf{21}      &		-0.39&	1.41&	-0.42	\\
\textbf{22}      &		-0.45&	1.37&	-0.44	\\
\textbf{23}      &		-0.55&	1.15&	0.10	\\
\textbf{24}      &		0.19&	0.90&	-0.41	\\
\textbf{25}      &		0.11&	0.61&	-0.06	\\
\textbf{26}      &		-0.37&	0.91&	0.10	\\
\textbf{27}      &		0.31&	1.35&	0.11	\\
\textbf{28}      &		0.74&	0.81&	0.36	\\
\textbf{29}      &		-0.19&	0.40&	0.41	\\
\textbf{30}      &		0.08&	1.04&	-0.23	\\
\textbf{31}      &		-0.14&	0.92&	-0.06	\\
\textbf{32}      &		-0.14&	0.49&	-0.24	\\
\hline					
\textbf{33}      &		0.33&	-0.14&	-0.23	\\
\textbf{34}      &		1.94&	-0.11&	0.36	\\
\textbf{35}      &		2.58&	-0.13&	-0.10	\\
\textbf{36}      &		1.59&	0.11&	-0.11	\\
\textbf{37}      &		1.75&	-0.22&	-0.40	\\
\textbf{38}      &		1.28&	0.04&	0.04	\\
\textbf{39}      &		1.55&	-0.55&	-0.39	\\
\textbf{40}      &		0.86&	0.24&	-0.51	\\
\hline
\end{tabular}
\quad
  \begin{tabular}{c|ccc}
    \hline
   \multicolumn{1}{c|}{Item } & \multicolumn{1}{c}{F1} & \multicolumn{1}{c}{F2} & \multicolumn{1}{c}{F3} \\
\hline
\textbf{41}       &		1.41&	-0.13&	-0.08	\\
\textbf{42}       &		2.02&	0.31&	-0.11	\\
\textbf{43}       &		0.78&	-0.05&	-0.09	\\
\textbf{44}       &		2.24&	-0.72&	0.06	\\
\textbf{45}       &		0.87&	0.76&	0.50	\\
\textbf{46}       &		0.49&	0.01&	-0.07	\\
\textbf{47}       &		0.72&	0.34&	-0.34	\\
\textbf{48}       &		0.74&	0.54&	0.13	\\
\textbf{49}       &		0.39&	-0.02&	0.10	\\
\textbf{50}       &		1.85&	0.06&	-0.14	\\
\textbf{51}       &		1.50&	-0.06&	-0.01	\\
\textbf{52}       &		2.36&	0.07&	0.06	\\
\textbf{53}       &		2.45&	0.46&	-0.15	\\
\textbf{54}       &		1.44&	-0.14&	-0.01	\\
\textbf{55}       &		1.92&	0.09&	-0.09	\\
\hline					
\textbf{56}       &		0.17&	0.48&	1.76	\\
\textbf{57}       &		0.01&	0.25&	1.31	\\
\textbf{58}       &		0.16&	-0.86&	1.64	\\
\textbf{59}       &		-0.09&	0.43&	1.34	\\
\textbf{60}       &		-0.10&	-0.22&	1.45	\\
\textbf{61}       &		-0.08&	0.23&	1.64	\\
\textbf{62}       &		0.24&	-0.09&	1.35	\\
\textbf{63}       &		-0.46&	-0.02&	1.61	\\
\textbf{64}       &		-0.34&	-0.52&	1.99	\\
\textbf{65}       &		-0.08&	-0.25&	1.17	\\
\textbf{66}       &		-0.30&	0.30&	1.54	\\
\textbf{67}       &		0.09&	-0.11&	0.78	\\
\textbf{68}       &		0.01&	0.24&	0.55	\\
\textbf{69}       &		0.04&	0.28&	1.08	\\
\textbf{70}       &		-0.11&	0.42&	1.14	\\
\textbf{71}       &		0.57&	-0.18&	0.99	\\
\textbf{72}       &		-0.19&	0.56&	0.66	\\
\textbf{73}       &		-0.33&	-0.54&	1.71	\\
\textbf{74}       &		-0.28&	0.15&	1.59	\\
\textbf{75}       &		-0.25&	0.46&	1.29	\\
\textbf{76}       &		-0.29&	-0.41&	1.46	\\
\textbf{77}       &		0.62&	0.26&	1.24	\\
\textbf{78}       &		0.07&	0.13&	1.10	\\
\textbf{79}       &		0.19&	0.25&	0.93	\\
&&&\\
\hline
\end{tabular}
\caption{Estimated loading matrix for the three-factor model  based on the marginal likelihood.  The results are obtained after applying the oblimin rotation.  }\label{tab:3factor_full_mar}
\end{table}

\begin{table}[h]
  \centering
  \begin{tabular}{cccc}
    \hline
  &F1 &F2 & F3\\
  \hline
F1&  1.00& 0.03& -0.16\\
F2&  0.03& 1.00&  0.06\\
F3& -0.16& 0.06&  1.00\\
    \hline
  \end{tabular}
  \caption{The estimated covariance matrix for the latent factors based on the marginal likelihood.  The results are obtained after applying the oblimin rotation.  }\label{tab:corr_3factor_mar}
\end{table}

%
%Any appendices appear after the acknowledgement but before the references, and have titles. If there is more than one appendix, then they are numbered, as here.
%
%\begin{theorem}
%This is a rather dull theorem:
%\begin{equation}
%\label{A1}
%a + b = b + a;
%\end{equation}
%a little equation like this should only be displayed and labelled if it is referred to elsewhere.
%\end{theorem}
%
%\appendixtwo
%\section*{Appendix 2}
%\subsection*{Technical details}
%
%Often the appendices contain technical details of the main results.
%
%\begin{theorem}
%This is another theorem.
%\end{theorem}
%
%\appendixthree
%\section*{Appendix 3}
%
%Often the appendices contain technical details of the main results:
%\begin{equation}
%\label{C1}
%a + b = c.
%\end{equation}
%
%\begin{remark}
%This is a remark concerning equations~\eqref{A1} and \eqref{C1}.
%\end{remark}
\clearpage

\bibliographystyle{apalike}
\bibliography{paper-ref}

\end{document}